\documentclass[12pt, conference, final,
onecolumn,technote,a4paper]{IEEEtran}
\usepackage{latexsym}
\usepackage{mathrsfs}
\usepackage{multirow}
\usepackage{amsfonts,amssymb,amsmath,amsthm,bm}
\usepackage{color}
\usepackage{cite}
\usepackage{xspace}
\newcommand{\cproblem}[1]{\ensuremath{\mathsf{#1}}\xspace}
\newcommand{\LWE}{\cproblem{LWE}}
\newcommand{\BDD}{\cproblem{BDD}}
\newcommand{\SIS}{\cproblem{SIS}}
\newcommand{\SVP}{\cproblem{SVP}}
\newcommand{\DGS}{\cproblem{DGS}}
\newcommand{\GDP}{\cproblem{GDP}}
\newcommand{\SIVP}{\cproblem{SIVP}}
\newcommand{\CVP}{\cproblem{CVP}}
\newcommand{\LPN}{\cproblem{LPN}}
\newcommand{\piped}{\mathcal{P}}
\DeclareMathOperator{\poly}{poly}
\DeclareMathOperator{\dist}{dist}
\DeclareMathOperator{\id}{id}
\newcommand{\GRLWE}{\cproblem{GR}\mbox{-}\cproblem{LWE}}
\newcommand{\matA}{\ensuremath{\mathbf{A}}}
\newcommand{\matB}{\ensuremath{\mathbf{B}}}
\newcommand{\matC}{\ensuremath{\mathbf{C}}}

\newcommand{\matI}{\ensuremath{\mathbf{I}}}

\newcommand{\matM}{\ensuremath{\mathbf{M}}}

\newcommand{\LL}{\mathcal{L}}
\newcommand{\OO}{\mathcal{O}}
\newcommand{\M}{\mathcal{M}}
\newcommand{\Mat}{\operatorname{Mat}}
\newcommand{\Aut}{\operatorname{Aut}}
\newcommand{\I}{\mathcal{I}}
\newcommand{\B}{\mathcal{B}}

\newcommand{\veca}{\ensuremath{\mathbf{a}}}
\newcommand{\vecb}{\ensuremath{\mathbf{b}}}

\newcommand{\vece}{\ensuremath{\mathbf{e}}}

\newcommand{\vecr}{\ensuremath{\mathbf{r}}}
\newcommand{\vecs}{\ensuremath{\mathbf{s}}}
\newcommand{\vect}{\ensuremath{\mathbf{t}}}
\newcommand{\vecu}{\ensuremath{\mathbf{u}}}
\newcommand{\vecv}{\ensuremath{\mathbf{v}}}

\newcommand{\vecx}{\ensuremath{\mathbf{x}}}
\newcommand{\vecy}{\ensuremath{\mathbf{y}}}
\newcommand{\vecz}{\ensuremath{\mathbf{z}}}
\newcommand{\veczero}{\ensuremath{\mathbf{0}}}

\newcommand{\C}{\ensuremath{\mathbb{C}}}

\newcommand{\R}{\ensuremath{\mathbb{R}}}
\newcommand{\T}{\ensuremath{\mathbb{T}}}
\newcommand{\Z}{\ensuremath{\mathbb{Z}}}
\newcommand{\DLWE}{\cproblem{DLWE}}
\newtheorem{theorem}{Theorem}
\newtheorem{example}{Example}
\newtheorem{definition}{Definition}
\newtheorem{remark}{Remark}
\newtheorem{corollary}{Corollary}
\newtheorem{lemma}{Lemma}
\newtheorem{proposition}{Proposition}

\begin{document}

\title{Learning with Errors over Group Rings Constructed by Semi-direct Product$^\dag$}
\author{Jiaqi Liu, Fang-Wei Fu
\IEEEcompsocitemizethanks{\IEEEcompsocthanksitem Jiaqi Liu and Fang-Wei Fu are with Chern Institute of Mathematics and LPMC, Nankai University, Tianjin
300071, China, Emails: ljqi@mail.nankai.edu.cn, fwfu@nankai.edu.cn}
\thanks{$^\dag$This research is supported by the National Key Research and Development Program of China (Grant No. 2022YFA1005000), the National Natural Science Foundation of China (Grant Nos. 12141108 and 62371259), the Fundamental Research Funds for the Central Universities of China (Nankai University), and the Nankai Zhide Foundation.}
}

\maketitle

\begin{abstract}
 The Learning with Errors (\LWE) problem has been widely utilized as a foundation for numerous cryptographic tools over the years. In this study, we focus on an algebraic variant of the \LWE problem called \emph{Group ring} \LWE ($\GRLWE$). We select group rings (or their direct summands) that underlie specific families of finite groups constructed by taking the semi-direct product of two cyclic groups. Unlike the Ring-\LWE problem described in \cite{lyubashevsky2010ideal}, the multiplication operation in the group rings considered here is non-commutative. As an extension of Ring-$\LWE$, it maintains computational hardness and can be potentially applied in many cryptographic scenarios. In this paper, we present two polynomial-time quantum reductions. Firstly, we provide a quantum reduction from the worst-case shortest independent vectors problem (\SIVP) in ideal lattices with polynomial approximate factor to the search version of $\GRLWE$. This reduction requires that the underlying group ring possesses certain mild properties; Secondly, we present another quantum reduction for two types of group rings, where the worst-case \SIVP problem is directly reduced to the (average-case) decision $\GRLWE$ problem. The pseudorandomness of $\GRLWE$ samples guaranteed by this reduction can be consequently leveraged to construct semantically secure public-key cryptosystems.
\end{abstract}


\section{Introduction}
\label{intro}
\subsection{Lattice-based cryptography}
Lattice-based cryptography has been playing a significant role in post-quantum cryptography. It is well-known that traditional cryptographic systems such as RSA and Diffie-Hellman protocol rely on the computational hardness of certain number-theoretic problems. However, Shor~\cite{shor1994algorithms} introduced a quantum algorithm that can solve integer factorization and discrete logarithm problems in polynomial time, rendering these classical cryptosystems vulnerable to quantum attacks. In contrast, lattice-based cryptosystems have several advantages over such traditional cryptosystems. First, no efficient quantum algorithms have been discovered for lattice problems that are used to construct cryptographic tools for the moment. Consequently, efforts are currently underway to establish post-quantum cryptographic standards, with lattice-based schemes as leading candidates. Second, there is an inherent disparity between the average-case hardness required by cryptography and the worst-case hardness concerned in computational complexity theory. Unlike other types of cryptosystems, lattice-based cryptosystems own security guaranteed by the worst-case hardness of certain lattice problems when their parameters are suitably chosen. Ajtai \cite{ajtai1996generating} proposed the first public-key cryptosystem from the \emph{short integer solution} (\SIS) problem whose security is based on the worst-case hardness of well-studied lattice problems.  Third, lattice-based cryptography is computationally efficient in comparison with RSA and the Diffie-Hellman protocol, as it only requires linear operations on vectors or matrices modulo relatively small integers, without the need for large number exponentiations.

Regev \cite{regev2009lattices} was the first to demonstrate the hardness of the \emph{learning with errors} (\LWE) problem, which is an extension of \emph{learning parity with noise} (\LPN) problem. The \LWE problem has gained considerable attention in recent years, due to its strong hardness and practical applications in various fields, including cryptography, machine learning, and quantum computing. The \LWE problem can be formulated as follows: Let $q>\poly(n)$ be a positive integer, fix a secret vector $\vecs\in\Z_{q}^n$, sample arbitrarily many $\veca_i\in\Z_{q}^n$ independently from the uniform distribution,  where $1\leq i\leq m$, and let $e_i\in\Z_q$ be independent ``short" elements which are typically sampled from a discrete Gaussian distribution. Then, compute the \LWE samples as $$(\veca_i,b_i=\langle \vecs, \veca_i\rangle+e_i)\in\Z_q^n\times\Z_q,$$ where $\langle\cdot,\cdot\rangle$ denotes the Euclidean inner product in $\R^n$. The \LWE problem asks to recover the secret vector $\vecs$ given $\LWE$ samples (search version) or distinguish \LWE samples from uniformly random samples in $\Z_{q}^n\times \Z_q$ with a non-negligible advantage within polynomial time (decision version). 

In \cite{regev2009lattices}, Regev initially proved the hardness of the \LWE problem for certain parameters, assuming the hardness of finding short vectors in lattices. Specifically, Regev established a (quantum) reduction from worst-case (decisional) approximate $\SVP$ (\emph{shortest vector problem}) and $\SIVP$ (\emph{shortest independent vectors problem})  for $n$-dimensional lattices to the average-case \LWE problem with the same dimension $n$ as the corresponding lattices. Informally speaking, the reduction suggests that if someone could solve \LWE in polynomial time with non-negligible probability, it would imply the existence of a quantum algorithm capable of solving approximate $\SVP$ or $\SIVP$ in arbitrary lattices with the same dimension as in \LWE problem. Later, Peikert \cite{peikert2009public} presented a classical reduction from worst-case approximate $\SVP$, but with worse parameters than those in \cite{regev2009lattices}. 

While Regev's \LWE enjoys worst-case hardness, it suffers from drawbacks regarding efficiency. It scales poorly with increasing dimension $n$, which makes it impractical for high-dimensional settings. Inspired by the NTRU \cite{hoffstein1998ntru} cryptosystem proposed by Hoffstein et al., which is constructed from algebraic lattices, Stehl\'{e} et al.~\cite{stehle2009efficient} proposed a more efficient variant of the \LWE problem over a cyclotomic ring. Lyubashevsky et al.~\cite{lyubashevsky2010ideal} then presented a (quantum) worst-case-to-average-case reduction from ideal lattice problems in the ring of integers of certain number fields (notably cyclotomic/Galois ones) to the Ring-\LWE problem. The improved efficiency of this problem stems from the additional algebraic structures of the rings. The utilization of number-theoretic transformation (NTT), a variant of fast Fourier transformation (FFT) also accelerates the multiplication of ring elements in certain cyclotomic rings, contributing to its computational efficiency. Precisely, Ring-\LWE also considers equations of the form $b_{i}=a_i\cdot s + e_i$ for a secret element $s\in R^{\vee}/qR^{\vee}$ and public uniformly random $a_i\in R/qR$, where $q$ is the modulus, $R$ is a commutative ring typically chosen as the ring of algebraic integers of a number field and $R^{\vee}$ is the dual ideal of $R$. Although some algorithms are more efficient on ideal lattices than on general lattices due to their additional algebraic structures, none are currently known to threaten well-designed cryptographic constructions. In particular, no algorithms are known that significantly outperform the best general lattice algorithms in attacking such schemes. Later, Peikert et al.~\cite{peikert2017pseudorandomness} gave a direct quantum reduction from worst-case (ideal) lattice problems to the decision version of (Ring)-\LWE.  Brakerski et al.~\cite{brakerski2014leveled} introduced Module-\LWE, balancing the efficiency and compactness of Ring-\LWE and the hardness of Regev's plain \LWE. Module-\LWE has more flexible parameters than Ring-\LWE and Regev's standard \LWE, making it a reasonable choice for standardization. Many schemes submitted to NIST for post-quantum standardization such as CRYSTALS-Kyber, CRYSTALS-Dilithium, and Saber, rely on the assumed hardness of \LWE variants. Additionally, there are several other algebraically structured \LWE variants, e.g., Polynomial-\LWE~\cite{rosca2018ring}, Order-\LWE~\cite{bolboceanu2019order}, and Middle-Product \LWE~\cite{rocsca2017middle}. Peikert et al.~\cite{peikert2019algebraically} have developed a unified framework for all proposed \LWE variants (over commutative base rings), which encompasses all prior
“algebraic” \LWE variants defined over number fields. This work has led to more simplified and tighter reductions to Ring-\LWE.

\SIS and \LWE are two off-the-shelf problems possessing reductions from worst-case lattice problems. They are both widely used to construct powerful cryptographic primitives. \LWE and its variants have numerous applications in cryptography including secure public-key encryption~\cite{regev2009lattices,peikert2008framework,peikert2009public,lyubashevsky2010ideal,micciancio2012trapdoors}, key-homomorphism PRF~\cite{boneh2013key,banerjee2014new}, oblivious transfer~\cite{peikert2008framework}, identity-based encryption (IBE)~\cite{gentry2008trapdoors,cash2012bonsai,agrawal2010efficient}, attribute-based encryption (ABE)~\cite{agrawal2011functional,gorbunov2015attribute,boneh2014fully,gorbunov2015predicate}, fully homomorphic encryption (FHE)~\cite{gentry2009fully,gentry2013homomorphic,brakerski2014efficient}, and multilinear maps~\cite{garg2013candidate,gentry2015graph}.

\subsection{Our results}

Similar to the work in \cite{cheng2022lwe}, we consider the \LWE problems over group rings. Let $G=\{g_1,g_2,\ldots,g_n\}$ be a finite group, and let $R$ be a commutative ring, elements in the group ring $R[G]$ are formal sums $$\sum_{i=1}^{n}r_ig_i,\quad r_i\in R.$$ In this paper, two types of non-commutative finite groups are constructed, each of which is a semi-direct product of two cyclic groups (e.g. $\Z_m\ltimes\Z_{n}, \Z_{n}^{\ast}\ltimes\Z_{n}$). We use them as the underlying groups of the group rings.

To study the relationship between an algebraic structure and a lattice in $\R^n$, we need to embed the structure into the Euclidean space $\R^n$. For \LWE variants based on rings, there are two types of embeddings into $\R^n$: \emph{canonical embedding} and \emph{coefficient embedding}. Canonical embedding provides nice geometric and algebraic properties. One can refer to the work of Lyubashevsky et al. in \cite{lyubashevsky2010ideal}. They embed the base ring $R$ (such as the cyclotomic ring $\Z[x]/\langle x^n +1\rangle$) into a linear space $H\subset \R^{s_1}\times \C^{2s_2}$ defined as $$H=\{(x_1,x_2,\ldots,x_n)\in \R^{s_1}\times \C^{2s_2}: x_{s_1+s_2+i}=\overline{x_{s_1+s_2}} \text{ for } 1\leq i\leq s_2\}$$ where $s_1+2s_2=n$. Every element $a\in R$ is mapped into $H$ by taking $n$ distinct embeddings of $a$  ($s_1$ real embeddings and $2s_2$ complex embeddings).  It can be easily verified that $H$, with the Hermite inner product, is isomorphic to the Euclidean inner product space $\R^{n}$. Canonical embedding often leads to tighter error rates and enables to obtain more compact systems due to its component-wise multiplication. On the other hand, the coefficient embedding used in \cite{stehle2009efficient} directly maps the ring elements to real vectors according to its coefficients under a specific basis. In this paper, we primarily use coefficient embedding for the sake of simplicity, following the approach described in \cite{cheng2022lwe}. However, we also need to use an extended version of canonical embedding to analyze the generalized \LWE problem due to its connection to irreducible group representations.

Taking an example for the canonical embedding in \cite{lyubashevsky2010ideal}, the ring $\mathbb{C}[x]/\langle x^n +1\rangle$ can be viewed as a direct summand of $\C[C_{2n}]$ where $C_{2n}$ denotes the cyclic group of order $2n$. Precisely, \begin{align}\label{isomorphism}       
       \C[C_{2n}]\cong\C[x]/\langle x^{2n}-1\rangle&\cong \C[x]/\langle x^n + 1\rangle\oplus \C[x]/\langle x^n -1\rangle\cong\bigoplus_{i=0}^{2n-1}\C[x]/\left\langle x-\mathrm{e}^{2\pi i\sqrt{-1}/2n}\right\rangle.
\end{align} One can observe that each direct summand above corresponds to a  (one-dimensional) irreducible representation of $C_{2n}$. The canonical embedding of $\Z[x]/\langle x^n +1\rangle$ consists of a natural inclusion into $\C[x]/\langle x^n +1\rangle$ and an isomorphism obtained from \eqref{isomorphism}: $$\C[x]/\langle x^n +1\rangle\cong\bigoplus_{0\leq i\leq 2n-1, 2 \nmid i}\C[x]/\left\langle x-\mathrm{e}^{2\pi i\sqrt{-1}/2n}\right\rangle.$$ 

To generalize the canonical embedding to the group ring constructed from a non-commutative group, we have to deal with irreducible representations of which the dimensions are greater than 1. By applying the Artin-Wedderburn theorem \cite{serre1977linear}, we can decompose the group algebra $\C[G]$ into the direct sum of matrix algebra over $\C$ uniquely. Each direct summand, denoted by $\C^{d\times d}$ matches exactly an irreducible representation of dimension $d$ for the group $G$. The complicated form of this decomposition is one of the reasons to use coefficient embedding in this paper rather than canonical embedding.

In the case of a group ring $\Z[G]$, we can regard it as a free $\Z$-module. All the group elements form a $\Z$-basis of $\Z[G]$ inherently. An element $\sum_{i=1}^n r_ig_i\in\R[G]$ can be represented by an $n$-dimensional vector $(r_1,r_2,\ldots, r_n)\in\R^n$. Under this embedding, each ideal of $\Z[G]$ corresponds to a lattice in $\R^n$, which is called an \emph{ideal lattice}. In our paper, we mainly study the relationship between the hard problems in ideal lattices and the \LWE problem over group rings.

\subsubsection{Avoiding the potential attack}
For  ``hard" problems such as \SVP  and \SIVP in certain ideal lattices, several potential attacks have been discovered. Various quantum polynomial-time algorithms have been developed to solve these problems in specific principal ideal lattices as mentioned in \cite{biasse2016efficient,cramer2016recovering,pan2021ideal}. Some studies generalized the SVP algorithms to non-principal ideal lattices such as mentioned in \cite{cramer2017short}. However, it is important to note that the Ring-\LWE problem present in \cite{lyubashevsky2010ideal} is not under direct threat from the attacks above, as the corresponding lattice problem serves as a lower bound of the security. Regarding the \LWE problem over group rings studied in this paper, it may not be affected by these algorithms for some reasons. On the one hand, the ideals of the group rings may not be principal. On the other hand, these attacks are basically designed for rings with commutative operation. Therefore, the attack can not be directly applied to the scenarios discussed in this paper, where the multiplication operation is not commutative.

\subsubsection{Why selecting such group rings}

In this paper, we have chosen the quotient ring $\Z[\Z_{m}\ltimes\Z_{n}]/\langle t^{n/2} +1\rangle$ instead of $\Z[\Z_{m}\ltimes\Z_{n}]$ as a candidate, where $t$ is the generator of $\Z_n$, though two reductions for $\Z[\Z_{m}\ltimes\Z_{n}]$ are both correct. 
For Ring-\LWE, there are some attacks by mapping the ring elements into small order elements to extract useful information. For instance, in~\cite{elias2015provably,chen2017attacks}, they tell us that it is not secure to use $\Z[C_{2n}]\cong\Z[x]/\langle x^{2n}-1\rangle$ as the underlying ring due to its vulnerability arising from the mapping $\Z[x]/\langle x^{2n}-1\rangle\rightarrow \Z[x]/\langle x-1\rangle$, which may leak some information that allows for secret recovery. Nevertheless, we often select polynomial rings of the form $\Z[x]/\langle x^n+1\rangle$, a direct summand of $\Z[C_{2n}]$, to avoid the information leakage resulting from such attacks.

Additionally, when performing the multiplication of two elements in the group ring, it is necessary to exploit all irreducible representations of the finite group due to the coefficient embedding. This implies that the highest dimension of irreducible representation cannot be excessively large. It can be verified that the dimension of irreducible representations of $\Z_m\ltimes\Z_n$ (Type I) is no more than 2, and that the dimension of irreducible representations of $\Z_{n}^{\ast}\ltimes\Z_n$ (Type II) is no more than $n-n/p_1p_2\cdots p_s$, where $p_1,p_2,\ldots, p_s$ are all the distinct prime factors of $n$.

\subsubsection{Relationship with other \LWE variants}

\LWE over group rings involves a non-commutative multiplication operation, unlike the \LWE variants studied in \cite{peikert2019algebraically}, where commutative algebraic systems are mainly concerned. Consider a group ring $R[G]$, where $G=\{g_1,g_2,\ldots,g_n\}$ is a finite group of order $n$. In this settings,  $g_1,g_2,\ldots,g_n$ naturally form a $R$-basis of $R[G]$. Let $a,s\in R[G]$, then the product $a\cdot s$ can also be expressed as product of a matrix and a vector with entries from $R$. Precisely, the left multiplication determined by $a$ is an $R$-linear transformation over $\R^n$, then $a$ uniquely specifies a unique matrix from $\R^{n\times n}$ under the natural basis, denoted by $\M(a)$. Then the product $a\cdot s$ can be equivalently represented as the product of $\M(a)$ with the coefficient vector of $s$.

\LWE over group rings can also be viewed as a ``structured-module" \LWE. As stated in \cite{cheng2022lwe}, the group ring is selected as $\Z[D_{2n}]$, i.e., $\Z[\Z_{2}\ltimes \Z_n]$. Let $r,t$ be generators of $\Z_2$ and $\Z_n$, respectively. Let $s=s_{1}(t)+ s_2(t)\cdot r, a=a_1(t)+ a_2(t)\cdot r, b =b_1(t)+ b_2(t)\cdot r\in\Z[D_{2n}]$, where $s_1(t),s_2(t),a_1(t),a_2(t),b_1(t),b_2(t)\in\Z[t]/\langle t^n-1\rangle,$ then the equation $b\approx s\cdot a$ can be expressed in the matrix form as follows: $$\left(\begin{array}{c}
       b_1\\
       b_2       
\end{array}\right)\approx \matM\cdot\left(\begin{array}{c}
       s_1\\
       s_2       
\end{array}\right).$$ Here, $\matM$ is a structured $2\times 2$ matrix with elements in $\Z[t]/\langle t^n-1\rangle$ all of which can be computed from $a_1(t),a_2(t)$ efficiently. It is worth noting that the matrix $\matM$ is structured, since only two elements ($a_1(t)$ and $a_2(t)$) of the four need to be sampled randomly, while in the situation of unstructured-module \LWE, all entries in $\matM$ are sampled randomly as in \cite{langlois2015worst}. That's the reason why the \LWE over group rings can be viewed as a ``structured-module'' \LWE. Compared with the classical module \LWE, to produce the same number of pseudorandom \LWE samples, the additional structure can reduce the cost of 
uniform randomness for sampling $\matM$ by a factor of $1/2$.

In \cite{pedrouzo2015multivariate}, Pedrouzo et al. introduced an extension of the Ring-\LWE framework proposed in \cite{lyubashevsky2010ideal} by adding other indeterminates to the polynomial ring. Specifically, they generalized the univariate polynomials in ring $\Z[x]/\langle x^n+1\rangle$ to multivariate polynomials in ring $\Z[x_1,x_2,\ldots,x_m]/\langle x_1^n+1,\ldots, x^n_m+1\rangle$.   In particular, multivariate Ring-\LWE with two indeterminates can also be interpreted as an instance of \LWE over the group ring discussed here. In this setting, the underlying group is represented by the direct product of two cyclic groups. However, in this paper, we primarily focus on the semi-direct product of two cyclic groups. It should be emphasized that, in some sense, the \LWE considered in this paper can be regarded as two-variate Ring-\LWE, but the two variables are not algebraically independent.

There are several schemes that exploit non-commutative algebraic structures. The idea of \LWE problem over group rings originates from non-commutative NTRU schemes \cite{DBLP:journals/iacr/YasudaDS15} where the authors employ group rings $\Z[G]$ instead of $\Z[x]/\langle x^n-1\rangle$ considered in the original NTRU scheme~\cite{hoffstein1998ntru}. Grover et al.~\cite{grover2022non} showed another non-commutative \LWE variant over cyclic algebras and provided hardness proof by giving 
a reduction from hard lattice problems.

\subsubsection{Proof of reductions}
In this paper, we present two reductions for the $\GRLWE$ problem. Both reductions make full use of \emph{iterative steps} (Lemma \ref{iterative} and Lemma \ref{iterative2}) mentioned in \cite{regev2009lattices} and \cite{lyubashevsky2010ideal}. One can refer to Section \ref{sec:3} and Section \ref{sec:4} for detailed information.

The first reduction is from a worst-case lattice problem to the search $\GRLWE$ problem. The reduction is based on the work of \cite{lyubashevsky2010ideal} and \cite{peikert2019algebraically}. It can be applied to 
any group ring that satisfies certain mild properties. Precisely, the reduction applies to those group rings satisfying if the ideal of the group ring has an inverse ideal, then the inverse ideal and the dual ideal are essentially equivalent, up to a certain permutation of the coefficients. Since not every ideal of the group rings is invertible, this reduction only considers those invertible ones. With the help of the search $\GRLWE$ oracle, we are capable of sampling from discrete Gaussian distributions with narrower and narrower widths in (invertible) ideal lattices through iterative steps. This process continues until the Gaussian parameter reaches our expectations. Consequently, we can obtain ``short" vectors in the lattice except with negligible probability.

The second one is a direct reduction from a (worst-case) lattice problem to the (average-case) decision $\GRLWE$ problem. The procedure closely follows the approaches outlined in \cite{peikert2017pseudorandomness} and \cite{cheng2022lwe}. The coefficient embedding and 
generalized canonical embedding both play crucial roles in this reduction. We establish the hardness by providing the reduction from solving \emph{bounded distance decoding} (\BDD), the problem of finding the closest lattice vector to a target $\vect\in\R^n$ that is sufficiently close to a lattice $\mathcal{L}$, to the decision $\GRLWE$ problem. During the reduction process, the decision \LWE oracle behaves like an oracle with a ``hidden center" defined by \cite{peikert2017pseudorandomness}. 
To see how the oracle works for reduction, \cite{regev2009lattices} states that one can transform an \BDD instance into an \LWE sample whose secret corresponds to the closest lattice vector $\vecv \in \mathcal{L}$ to $\vect$. With a suitable oracle for decision \LWE and by incrementally moving $\vect$ towards the closest lattice vector $\vecv$ by carefully measuring the behavior of the oracle, we can detect the distance between the moving point $\vect$ and $\LL$ by monitoring the acceptance probability of the oracle, which is a monotonically decreasing function of $\dist(\vect,\LL)$. This makes it possible to solve \BDD by repeatedly perturbing $\vect$ to a new point $\vect^{\prime}$, and testing whether the new point is significantly closer to the lattice. For further details, we refer to Section \ref{sec:4}.

\subsection{Applications}

Due to the pseudorandomness of \LWE samples over group rings, it is possible to construct a public-key cryptosystem that provides semantic security. The security is based upon the hardness of the worst-case \SIVP problem with a polynomial approximate factor. We can combine the result of this paper with the cryptosystem present in  \cite{lyubashevsky2010ideal}, which is also similar to both Diffie-Hellman protocol~\cite{hellman1976new} and ElGamal protocol~\cite{elgamal1985public}. 

Informally, let $R$ be some group ring or a quotient ring of some group ring. For instance, we can select a ring $R=\Z[G]/\langle t^{n/2}+1\rangle$, where $G$ is a semi-direct product of cyclic groups $\Z_{m}=\langle s\rangle$ and $\Z_n=\langle t\rangle$. Let $q$ be a positive integer. The public-key cryptosystem based on $\GRLWE$ is as follows: The key-generation algorithm first generates a uniformly random element $a\in R_q:=R/qR$ (under coefficient embedding) and two ``short'' elements $s,e\in R_q$ from error distribution (typically selected from a discrete Gaussian distribution with ``narrow'' error). It outputs a $\GRLWE$ sample $(a,b=s\cdot a+e)\in R_q\times R_q$ as the public key and $s$ as the secret key. To encrypt a $|G|/2$-bit message $m\in\{0,1\}^{|G|/2}$ that can be viewed as a polynomial of degree $n/2-1$ with $0$ or $1$ coefficients, the encryption algorithm first generates three ``short'' random elements $r,e_1$ and $e_2\in R_q$ from error distribution and outputs $(u,v)\in R_q\times R_q$ with $u =  a\cdot r +e_1\bmod{q}$ and $v=b\cdot r+e_2+\lfloor q/2\rceil\cdot m\bmod{q}.$ To decrypt $(u,v)\in R_q\times R_q$, the decrypt algorithm computes $$\hat{m}=v-s\cdot u=(e\cdot r- s\cdot e_1+e_2)+\lfloor q/2\rceil\cdot m\bmod q.$$ By choosing appropriate parameters, the element $e\cdot r- s\cdot e_1+e_2$ can be restricted to be very ``short''. By rounding the coefficient of $\hat{m}$ to either $0$ or $\lfloor q/2\rceil$, we can efficiently recover the bit string except with negligible probability.

To understand why this cryptosystem is semantically secure, we first observe that the public key is generated as a $\GRLWE$ sample, though with ``short" secret $s$, which can also be proved to be pseudorandom, as the ``normal form \LWE'' discussed in \cite{applebaum2009fast}. Such a sample becomes indistinguishable from a truly random sample except with negligible probability. By replacing the public key with a randomly uniform sample (it makes negligible difference to this system), the encrypted message $(a,u)\in R_q\times R_q$ and $(b,v)\in R_q\times R_q$ both become exact $\GRLWE$ samples with ``random'' secret vector $r$. Therefore, these two encrypted messages are pseudorandom, which implies the semantic security of the system.

\subsection{Paper organization}
In Section 2, we give some basic knowledge of lattice and related computationally hard problems. In Section 3, we provide a reduction from worst-case lattice problems to the search version of $\GRLWE$ over some group rings with special properties. In Section 4, we present a direct reduction from worst-case lattice problems to the decision version of $\GRLWE$, considering two types of group rings. This generalizes the results of \cite{cheng2022lwe} and enables to construct some cryptographic primitives. In Section 5, we summarize our findings and offer some thoughts for further study.

\section{Preliminaries}
\label{sec:2}

First, we present some notations used throughout this paper.

For $x\in\R$, we denote $\lfloor x\rfloor$ as the largest integer not greater than $x$ and $\lfloor x\rceil:=\lfloor x+\frac{1}{2}\rfloor$ as the nearest integer to $x$ (if there exist two nearest integers, select the larger one). For integer $n\geq 2$, we denote $\Z_n=\Z/n\Z$ as the cyclic group of order $n$ with addition modulo $n$. Denote $\varphi(m)$ as the Euler's totient function of $m$ for any positive integer $m$, i.e., the number of positive integers that are coprime with $m$ and also not greater than $m$. 

Throughout this paper, we use bold lower-case letters to denote column vectors, e.g., $\vecx,\vecy,\vecz$. We use bold upper-case letters to denote matrices, e.g., $\matA,\matB,\matC$. The transpose of vector $\vecx$ and matrix $\matA$ is denoted as $\vecx^{t}$ and $\matA^{t}$, respectively. The rounding function mentioned above can be applied element-wise to vectors, e.g. $\lfloor \vecu\rceil$ round each entry of $\vecu$ to its nearest integer. The inner product of two vectors $\vecu,\vecv$ is denoted by $\langle \vecu, \vecv\rangle$.  

If $S$ is a nonempty finite set, then $x\leftarrow S$ denotes that $x$ is a random variable uniformly distributed over $S$. If $\varphi$ is a probability density function, then $x\leftarrow \varphi$ denotes that $x$ is a random variable distributed as $\varphi$. Given two probability distribution functions $\rho_1,\rho_2$ over $\R^n$, the \emph{statistical distance} between $\rho_1$ and $\rho_2$ is defined as $$\Delta(\rho_1,\rho_2):=\frac{1}{2}\int_{\R^n}|\rho_1(\vecx)-\rho_2(\vecx)|d\vecx.$$

\subsection{Lattice Background}

In this section, we introduce some definitions and discuss some related ``hard'' problems regarding lattices.

A \emph{lattice} is defined as a discrete additive subgroup of $\mathbb{R}^n$. Equivalently, an $n$-dimensional lattice $\LL$ is a set of all integer combinations of $n$ linearly independent basis column vectors $\matB:=(\vecb_1,\vecb_2,\ldots,\vecb_n)\in\R^{n\times n}$, i.e., $$\LL=\LL(\matB):=\matB\cdot \mathbb{Z}^n=\left\{\sum_{i=1}^n z_i\vecb_i: z_i\in\Z, 1\leq i\leq n\right\}.$$
Then lattice $\LL$ can be viewed as a full-rank free $\mathbb{Z}$-module with basis $\vecb_1,\vecb_2,\ldots,\vecb_n.$

Since a lattice $\LL$ is an additive subgroup of $\R^n$, we can obtain quotient group $\R^n/\LL$ (with operation induced by addition from $\R^n$), which has cosets $$\vecx+\LL=\{\vecx+\vecy:\vecy\in\LL\},\quad\text{for all}\ \vecx\in\R^n$$ as its elements.

The \emph{fundamental parallelepiped} of a lattice $\LL$ is defined as $$\piped(\matB):=\matB\cdot[0,1)^n=\left\{\sum_{i=1}^n c_i\vecb_i: 0\leq c_i< 1, 1\leq i\leq n\right\}.$$ It is clear that the definition of the fundamental parallelepiped depends on the choice of basis. Fixing a basis $\matB$, every coset $\vecx+\LL\in\R^n/\LL$ has a unique representative in $\piped(\matB)$. In fact, coset $\vecx+\LL$ has representative $\vecx-\matB\cdot\lfloor\matB^{-1}\cdot\vecx\rfloor\in\piped(\matB)$, which we denote by $\vecx\bmod \LL$. 

The \emph{determinant} of a lattice $\LL$ is the absolute value of the determinant of a basis $\matB$\footnote{It can be proved that the value of determinant is independent from the selection of the basis $\matB$.}, i.e., $$\det\LL:=|\det(\matB)|.$$

The \emph{minimum distance} of a lattice $\LL$ in a given norm $\|\cdot\|$ is the length of a shortest nonzero lattice vector, i.e.,$$\lambda_1(\LL):=\min\limits_{\veczero\neq\vecx\in\LL}\|\vecx\|.$$ The \emph{$i$-th successive minimum} $\lambda_i(\LL)$ is defined as the smallest $r\in\R$ such that $\LL$ has $i$ linearly independent vectors with norm no greater than $r$. In this paper, we use $\ell_2$-norm unless otherwise specified. 

The \emph{dual lattice} of $\LL$ is defined as $$\LL^{\vee}:=\left\{\vecx\in\R^n:\langle\vecx,\vecy\rangle\in\Z,\text{ for all }\vecy\in\LL\right\},$$ where $\langle \cdot, \cdot\rangle$ represents inner product (Euclidean inner product by default). And it is easy to prove that $\LL^{\vee}$ is also a lattice in $\R^n$ and if $\matB$ is a basis of $\LL$, then $\matB^{-t}:=(\matB^{-1})^t=(\matB^t)^{-1}$ is a basis of $\LL^{\vee}$, hence $(\LL^{\vee})^{\vee}=\LL.$

\subsection{Gaussian measures}
For any (column) vector $\vecu\in\R^n$ and definite positive matrix $\boldsymbol{\Sigma}\in\R^{n\times   n}$, we define the \emph{Gaussian distribution} with mean vector $\vecu$ and covariance matrix $\frac{1}{2\pi}\boldsymbol{\Sigma}\in\R^{n\times n}$ as distribution with the probability density function $$D_{\vecu,\boldsymbol{\Sigma}}(\vecx)=\frac{1}{\sqrt{\det\boldsymbol{\Sigma}}}\exp(-\pi(\vecx-\vecu)^t\boldsymbol{\Sigma}^{-1}(\vecx-\vecu))$$ for any $\vecx\in\R^n$. If $\vecu$ is a zero vector, we denote $D_{\vecu,\boldsymbol{\Sigma}}$ by $D_{\boldsymbol{\Sigma}}$ for short.

In particular, if $\vecu$ is a zero vector and $\boldsymbol{\Sigma}$ is a diagonal matrix with diagonal entries $r_1^2,r_2^2,\ldots,r_n^2$ $\in\R^{+}$, then $D_{\vecu,\boldsymbol{\Sigma}}$ degrades to an elliptical (non-spherical) Gaussian distribution, which we denote $D_{\vecr}$ for convenience, where $\vecr=(r_1,r_2,\ldots,r_n)\in\R^n$. Furthermore, if $r_1=r_2=\ldots=r_n=r$, then this distribution degrades to a spherical Gaussian distribution with probability density function $\frac{1}{\sqrt{\det{\boldsymbol{\Sigma}}}}\exp(-\pi\|\vecx\|^2/r^2)$, which we denote $D_{r}$ for short. These functions can be extended to sets in the usual way, e.g., $D_{\vecu,\boldsymbol{\Sigma}}(S)=\sum_{\vecx\in S}D_{\vecu,\boldsymbol{\Sigma}}(\vecx)$, where $S\subseteq \R^n$ is a countable set. 

For an $n$-dimensional lattice $\LL$, a vector $\vecu\in\R^n$, and a real $r>0$, we define \emph{discrete Gaussian probability distribution over} the coset $\LL+\vecu$ with parameter $\boldsymbol{\Sigma}\in\R^{n\times n}$ as
$$D_{\LL+\vecu,\boldsymbol{\Sigma}}(\vecx):=\frac{D_{\boldsymbol{\Sigma}}(\vecx)}{D_{\boldsymbol{\Sigma}}(\LL+\vecu)},\quad \text{for all}\ \vecx\in\LL+\vecu.$$
Note that the support set of $D_{\LL+\vecu,\boldsymbol{\Sigma}}$ is $\LL+\vecu$, a discrete set in $\R^n$.

The (continuous or discrete) Gaussian distribution owns numerous helpful properties. We list some of them as follows:

\begin{lemma}[\cite{regev2009lattices}, Claim 2.2]\label{Gaussianbound}
       For $0<\alpha<\beta\leq 2\alpha$, the statistical distance between $D_{\alpha}$ and $D_{\beta}$ is no greater than $10(\beta/\alpha-1)$.       
\end{lemma}

As in the continuous version, almost all samples from $n$-dimensional discrete Gaussian distribution can be bounded in a sphere with a radius factor of $\sqrt{n}$. 

\begin{lemma}[\cite{banaszczyk1993new}, Lemma 1.5 (i)]\label{boundgaussian}
       For any $n$-dimensional lattice $\LL$ and real $r>0$, a point sampled from $D_{\LL,r}$ has $\ell_2$ norm at most $r\sqrt{n}$, except with probability at most $2^{-2n}$.
\end{lemma}

Micciancio et al.~\cite{micciancio2007worst} first introduced \emph{smoothing parameter} of lattices, which can be used to characterize the similarity between discrete Gaussian distribution over the lattice and continuous Gaussian distribution with the same parameter. However, for this paper, we require a more generalized definition mentioned in \cite{cheng2022lwe} as below. 

\begin{definition}[Smoothing parameter,\cite{cheng2022lwe}, \cite{micciancio2007worst}] Let $\LL$ be a lattice in $\R^n$ and $\LL^{\vee}$ be its dual lattice. For a real $\varepsilon>0$, the \emph{smoothing parameter} $\eta_{\varepsilon}(\LL)$ is defined to be the smallest $s$ such that $$\sum_{\vecy\in\LL^{\vee}\backslash\{0\}}D_{1/s}(\vecy)=\sum_{\vecy\in\LL^{\vee}\backslash\{0\}}\exp(-\pi s^2\cdot\|\vecy\|^2)\leq\varepsilon.$$

For any $\matA\in\R^{n\times n}$, $\matA$ is defined to satisfy \emph{smoothing condition} if $$\sum_{\vecy\in\LL^{\vee}\backslash\{0\}}\exp(-\pi \cdot\vecy^{t}\matA\matA^{t}\vecy)\leq\varepsilon,$$ which we denote as $\matA\geq\eta_{\varepsilon}(\LL)$.     
\end{definition}

The following lemma states an important property of the smoothing parameter and also explains the term ``smoothing'' to some extent. Suppose we have a vector $\vecv\in\R^n$ sampling from Gaussian distribution $D_{r}$ where $r$ exceeds the smoothing parameter of a certain lattice $\LL$, then $\vecv\bmod{\LL}$ is approximately uniform distributed over any fundamental parallelepiped $\mathcal{P}$ except with small statistical distance. Since the volume of any fundamental parallelepiped $\mathcal{P}$ of $\LL$ is equal to $\det(\LL)$, then the uniform distribution over $\mathcal{P}$ has probability density function $1/\det(\LL)$. To summarize, the results can be stated as follows.

\begin{lemma}[\cite{micciancio2007worst}, Lemma 4.1]
       For any $n$-dimensional lattice $\LL$, $\varepsilon>0$ and $r\geq\eta_{\varepsilon}(\LL)$, the statistical distance between $D_r\bmod\LL$ and the uniform distribution over $\R^{n}/\LL$ is at most $\varepsilon/2$, then we have $$D_{r}(\LL+\vecu)\in(1\pm\varepsilon)\frac{1}{\det(\LL)},$$ where $\vecu$ is an arbitrary vector in $\R^n$.     
       \end{lemma}
       
       The following lemma states two useful bounds about smoothing parameters.
\begin{lemma}[\cite{micciancio2007worst}, Lemma 3.2, 3.3]\label{smooth}
              For any $n$-dimensional lattice $\LL$, we have $\eta_{2^{-2n}}(\LL)\leq\sqrt{n}/\lambda_{1}(\LL^{\vee})$, and $\eta_{\varepsilon}(\LL)\leq\sqrt{\ln(n/\varepsilon)}\lambda_{n}(\LL)$ for all $0<\varepsilon<1$.
\end{lemma}

\subsection{Representations of finite groups}
Let $G$ be a finite group, a pair $(V,\rho)$ ($V$ for short) is called a \emph{representation} of $G$ if $V$ is a linear space and $\rho$ is a group homomorphism that maps $G$ to $\mathrm{GL}(V)$, where $\mathrm{GL}(V)$ is the group of all invertible linear transformations over $V$. If a linear subspace $U$ of $V$ is preserved by the action of any element in $G$, i.e., $\rho(g)u\in U$ for all $g\in G, u\in U$, then $(U,\rho)$ ($U$ for short) is called a \emph{subrepresentation} of $V$. It follows immediately that every representation $(V,\rho)$ has two \emph{trivial subrepresentations}: $\{0\}$ and $V$ itself. $V$ is called \emph{irreducible} if $V$ has no subrepresentations other than the trivial subrepresentations. Otherwise, $V$ is called \emph{reducible}. The dimension of $V$ is called the \emph{dimension} of the representation.

\subsection{Group ring}
We follow the formulation of this definition as presented in~\cite{cheng2022lwe}. Let $G=\{g_1,g_2,\cdots,g_n\}$ be a finite group of order $n$, and let $R$ be a commutative ring with identity. The group ring $R[G]$ is the set of all finite formal sums $$\sum\limits_{i=1}^n r_ig_i,\quad r_i\in R.$$
Addition is defined coefficient-wise $$\sum\limits_{i=1}^n a_ig_i+\sum\limits_{i=1}^nb_ig_i=\sum\limits_{i=1}^n(a_i+b_i)g_i.$$ Multiplication is defined by extending the group operations $R$-bilinearly: $$\left(\sum\limits_{i=1}^na_ig_i\right)\left(\sum\limits_{i=1}^nb_ig_i\right)=\sum\limits_{r=1}^n\left(\sum\limits_{(i,j):\ g_ig_j=g_r}a_ib_j\right)g_r.$$
In particular, $R[G]$ is a free $R$-module of rank $|G|=n$ in a natural way, with $\{g_1, g_2, \ldots, g_n\}$ as an $R$-basis.

Next, we need to define the matrix form of the elements in a group ring. For any element $\mathfrak{h}=\sum_{i=1}^n a_ig_i\in R[G]$, it defines a linear transformation on $R[G]$, by its left multiplication law. We denote the $n\times n$ transformation matrix corresponding to $\mathfrak{h}$ with respect to the basis $\{g_1,g_2,\ldots,g_n\}$\footnote{The choice of the basis doesn't significantly affect the analysis in the following essentially, we provide this definition here to avoid ambiguity.} by $\M(\mathfrak{h})$. Define the matrix-norm $\|\mathfrak{h}\|_{\Mat}$ of $\mathfrak{h}$ as the spectral norm of the matrix $\M(\mathfrak{h})$, i.e., the square root of the largest eigenvalue of $\mathcal{M}(\mathfrak{h})\mathcal{M}(\mathfrak{h})^{T}$.

\subsection{Semi-direct product}
In this paper, we construct some finite groups mainly by taking semi-direct products on two cyclic groups.

\begin{definition}\label{semi}
Let $G$ and $H$ be two groups, and $\varphi: G\rightarrow \Aut(H)$ be a group homomorphism, where $\Aut(H)$ represents the automorphism group of $H$. The \emph{semi-direct product} group of $G$ and $H$ induced by $\varphi$ is defined as a set $K=\{(g,h)\mid g\in G,h\in H\}$  with multiplication law ``$\cdot$":$$(g_1,h_1)\cdot (g_2,h_2)=(g_1g_2,\varphi(g_2^{-1})(h_1)\cdot h_2).$$ We denote $K$ with multiplication ``$\cdot$" by $G\ltimes_\varphi H$, or simply $G\ltimes H$.
\end{definition}

It should be emphasized that semi-direct product as a binary operation of groups is not commutative. From the definition provided above, we can verify that $H$ is a normal subgroup of $K$, but $G$ may not be. Additionally, the semi-direct product $G\ltimes H$ may not be unique, as its structure varies depending on $\varphi$. In particular, if we consider the identity mapping for $\varphi$, i.e., for every element $g\in G$, $\varphi(g)=\id_H$, then the semi-direct product degrades to direct product, namely $G\times H$. In this case, the elements of $G$ and $H$ act ``independently" of each other.

\begin{example}\label{dihedral}
The dihedral group $D_{2n}$ can be generated by two elements $s$ and $t$ with the relation $sts=t^{-1}$, where the order of $s$ and $t$ are $2$ and $n$, respectively. Precisely, $$D_{2n}:=\langle s,t:s^2=t^n=1,sts=t^{-1}\rangle.$$ Then $D_{2n}$ can also be represent as $\mathbb{Z}_2\ltimes_\varphi\mathbb{Z}_n$, where $\varphi(s)=\operatorname{inv}$, where $\operatorname{inv}$ maps each element of $\Z_n$ to its inverse (under addition modulo $n$). 
\end{example}

Inspired by Example \ref{dihedral}, we first give two types of groups which are utilized in this paper.

\textbf{Type I}: $\Z_{m}\ltimes_{\varphi}\Z_{n}$: Let $m$ be an even positive integer, and $n$ be an arbitrary positive integer. Denote $\Z_m$ and $\Z_n$ as two cyclic groups with generators $s$ and $t$, respectively, i.e., $\Z_{m}=\langle s\rangle$ and $\Z_{n}=\langle t\rangle$. Let $\sigma$ be an element in the automorphism group $\operatorname{Aut}(H)$, which maps each element of $H$ to its inverse. Define 
a homomorphism $\varphi: G\rightarrow \operatorname{Aut} H$ with $\varphi(s)=\sigma$. Induced by $\varphi$, we can define $\Z_{m}\ltimes_{\varphi}\Z_{n}$, denoted as $\Z_{m}\ltimes\Z_{n}$ below: 
\begin{align}\label{typei}         \Z_{m}\ltimes\Z_{n}:=\langle s,t:s^m=1,t^n=1, sts^{-1}=t^{-1}\rangle.       
\end{align}

\textbf{Type II}: $\Z_{n}^{\ast}\ltimes_{\psi}\Z_{n}$: Let $n$ be an arbitrary positive integer and let $\Z_n$ denote a multiplicative cyclic group with generator $g$. Denote $\mathbb{Z}_{n}^{\ast}$ as the set of all positive integers no more than $n$ that are coprime with $n$. For instance, $\mathbb{Z}_{10}^{\ast}=\{1,3,7,9\}$. It is easy to verify that $\Z_n^{\ast}$ forms a group under the operation of multiplication modulo $n$.  
       
It is also known that the automorphism group of $\Z_n$ is isomorphic to $\Z_{n}^{\ast}$. Thus we have a natural isomomorphism $\psi$ mapping $\mathbb{Z}_{n}^{\ast}$ to $\operatorname{Aut}(\Z_n)\cong \mathbb{Z}_{n}^{\ast}$, where $\psi(a)=\psi_a\in \operatorname{Aut}(\Z_n)$ and $$\psi_a(g)=g^a.$$ Similar to Type I, the semi-direct product of $\Z_{n}^{\ast}$ and $\Z_n$ can be induced by $\psi$. Here we use the symbol $a\odot g^k$ to represent the element $(a,g^k)$ in the group $\Z_{n}^{\ast}\ltimes\Z_{n}$. Then the semi-direct product can be expressed as follows: \begin{align}\label{typeii}
\Z_{n}^{\ast}\ltimes\Z_{n}:=\{a\odot g^k: a\in\Z_{n}^{\ast}, 0\leq k<n\},
\end{align} 
the multiplication of which is induced by $\psi$ and satisfies $$(a\odot g^{k_1})(b\odot g^{k_2})=ab\odot g^{k_1b^{-1}+k_2}$$ as in Definition \ref{semi}.

\begin{remark}
     In particular, when $m=2$ is additionally satisfied in Type I, the group $\Z_m\ltimes\Z_n$ defined by \eqref{typei} is exactly the dihedral group $D_{2n}=\Z_{2}\ltimes\Z_{n}$.
\end{remark}

We select these two types of groups mainly because their irreducible representations can be computed by the representation theory of finite groups. We provide two lemmas in this paper to present the irreducible representations of these two types of groups. The proofs require some basic knowledge and techniques in representation theory. They are attached in the appendix section in details.

\begin{lemma}[Irreducible representations of Type I]\label{lemtypei}
       Let $G_1=\Z_m\ltimes_{\varphi}\Z_n=\langle s\rangle\ltimes_{\varphi}\langle t\rangle$ be a finite group defined by \eqref{typei}, where $m=2r$ is even and let $\omega$ and $\xi$ be the $m$-th and $n$-th primitive roots of unity, respectively.
\begin{enumerate}
       \item [(i)] If $n=2u+1$ is odd, then $G_1$ has $m=2r$ (non-equivalent) 1-dimensional irreducible representations $\chi_{i}$ $(i=1,2,\ldots, 2r)$, satisfying $$\chi_{i}(t)=1,\quad \chi_{i}(s)=\omega^{i};$$ and $ru$ (non-equivalent) 2-dimensional irreducible representations $\rho_{i,j}$ $(i=1,2,\ldots,r$, $j=1,2,\ldots,u)$, satisfying \begin{align*}
              &\rho_{i,j}(t)=\left(\begin{array}{cc}
                 \xi^{j}&\\
                 &\xi^{n-j}    
              \end{array}\right),\\
              &\rho_{i,j}(s)=\left(\begin{array}{cc}
                     0&\omega^{2i}\\
                     1&0    
                  \end{array}\right).
       \end{align*}
       \item [(ii)] If $n=2u$ is even, then $G_1$ has $2m=4r$ (non-equivalent) 1-dimensional irreducible representations $\chi_i$  ($i=1,2,\ldots,4r$), satisfying 
       \begin{align*}
              &\chi_{i}(t)=1,\quad\chi_{i}(s)=\omega^i,\ i=1,\ldots,2r,\\
              &\chi_{i}(t)=-1,\quad\chi_{i}(s)=\omega^i,\ i=2r+1,2r+2,\ldots,4r,
       \end{align*} and $r(u-1)$ (non-equivalent) 2-dimensional irreducible representations $(i=1,2,\ldots,r$, $j=1,2,\ldots,u-1)$ satisfying \begin{align*}
              &\rho_{i,j}(t)=\left(\begin{array}{cc}
                 \xi^{j}&\\
                 &\xi^{n-j}    
              \end{array}\right),\\
              &\rho_{i,j}(s)=\left(\begin{array}{cc}
                     0&\omega^{2i}\\
                     1&0    
                  \end{array}\right),
       \end{align*}
\end{enumerate} 
\end{lemma}

As for finite groups of Type II, we only provide an upper bound of the dimensions of all irreducible representations for simplicity, which is sufficient for the analysis in this paper. It can be verified that if $r$ is coprime with $s$, then $\Z_{rs}\cong \Z_r\times\Z_s$ and $\Z_{rs}^*\cong\Z_r^*\times\Z_s^*$. 
Furthermore, $\Z_{rs}^{\ast}\ltimes\Z_{rs}\cong(\Z_r^*\times\Z_s^*)\ltimes(\Z_r\times\Z_s)\cong (\Z_r^*\ltimes\Z_r)\times (\Z_s^*\ltimes\Z_s)$,
hence it suffices to list all the irreducible representations of $\Z_{p^k}^{\ast}\ltimes\Z_{p^k}$ where $p$ is a prime and $k$ is a positive integer. By tensoring the irreducible representations of $\Z_{p_i^{k_i}}^*\ltimes\Z_{p_i^{k_i}}$ for each prime factor $p_i$ in the decomposition  $n=p_1^{k_1}p_2^{k_2}\cdots p_m^{k_m}$, where $i=1,2,\cdots,m$, we can obtain all irreducible representations of $\Z_{n}^*\ltimes\Z_n$. For groups constructed by the semi-direct product of two Abelian groups, we can determine all their irreducible representations. Refer to \cite{serre1977linear} for detailed information.

\begin{lemma}[Irreducible representations of Type II]\label{Irreducible dimensions of Type II}
       Let $G_2=\Z_{p^k}^{\ast}\ltimes \Z_{p^k}$ is defined as \eqref{typeii}, where $p$ is prime and $k$ is a positive integer. Then $G_2$ only has irreducible representations with dimension no greater than $p^{k-1}$.
\end{lemma}
\begin{proof}
Let $n_1,n_2,\ldots,n_s$ be the dimensions of all (non-equivalent) irreducible representations of $G_2$. According to the representation theory, we have \begin{align*}
      & n_1^2+n_2^2+\cdots+n_s^2=|G_2|=p^{2k-1}(p-1),\\
       & n_i\mid |G_2| \quad \text{for\ } 1\leq i\leq s.
\end{align*}
       Thus it follows immediately that $n_i\leq p^{k-1} $ for all $i$.
\end{proof}

It is also crucial to analyze all the eigenvalues of the group ring elements when regarded as a linear transformation induced by the left multiplication law. In fact, we can compute all eigenvalues of the linear transformations determined by the elements in $\C[G]$, where $G$ belongs to Type I or Type II. To perform these computations, we need to use some fundamental representation theory and techniques.

\begin{lemma}\label{eigen}
       \begin{enumerate}
                     \item  Let $m$ be an even positive integer and $n$ be an arbitrary positive integer. Consider the group $G_1:=\mathbb{Z}_m\ltimes_{\varphi} \mathbb{Z}_n=\langle s\rangle\ltimes_{\varphi}\langle t\rangle$  defined by \eqref{typei}. Let $\mathfrak{h}=\sum_{i=0}^{m-1}s^{i}f_{i}(t)\in \C[G]$, where $f_i(t)$ is a polynomial of degree no more than $n$. All the eigenvalues of the matrix $\M(\mathfrak{h})$ are given by $$f_{0}(\xi^j)+\omega^i f_{1}(\xi^j)+\cdots+\omega^{i(m-1)} f_{m-1}(\xi^{j}),\quad i=0,1,\ldots,m-1, j=0,1,\ldots,n-1,$$ where $\omega,\xi$ are $m$-th,$n$-th primitive roots of unity, respectively.
       
       \item Let $p$ be a prime and $k$ be a positive integers. Denote $m:=p^k-p^{k-1}$ and consider the group $\Z_{p^k}^{\ast}\ltimes_{\psi}\Z_{p^k}$ defined by \eqref{typeii}. Let $a$ be a multiplicative primitive element of $\Z_{p^k}^{\ast}$. Let $\mathfrak{h}=\sum_{r=0}^{m-1}\sum_{s=0}^{p^k-1}f_{rs}(a^r\odot g^s)$, where $f_{rs}$ are complex numbers, then all the eigenvalues of the matrix $\M(\mathfrak{h})$ are given by $$\sum_{r=0}^{m-1}\sum_{s=0}^{p^k-1}f_{rs}(\omega^{ir}\cdot \xi^{js}),\quad i=0,1,\ldots,m-1,j = 0,1,\ldots,p^k-1,$$ where $\omega,\xi$ are $m$-th, $p^k$-th primitive roots of unity, respectively.
              \end{enumerate}
       \end{lemma}
       
\begin{proof}
It is natural that each element in $\C[G]$ determines a left multiplication action on $\C[G]$ itself. Here we consider the groups of Type I for instance, and the analysis for Type II follows essentially the same approach. Each irreducible subrepresentation of group $G_1$ can be found as a summand of left regular representation. Referring to Lemma \ref{lemtypei} and representation theory, the matrix determined by the regular representation over a specific basis is similar to a block diagonal matrix with block size no greater than $2$, where each block corresponds to an irreducible subrepresentation of $G$.  By Lemma \ref{lemtypei}, we know all irreducible subrepresentations of $G$. Therefore, all the eigenvalues can be calculated by considering the eigenvalues of each subrepresentation.
       \end{proof}

\subsection{Ideal lattice and coefficient embedding}
A \emph{left (integral) ideal} $\mathcal{I}$ of the group ring $\mathbb{Z}[G]$ is an additive subgroup of $G$ that is closed under left multiplication by any element in $\mathbb{Z}[G]$, i.e., $\mathfrak{h}x\in\I$ holds for any $\mathfrak{h}\in \mathbb{Z}[G]$ and $x\in \I$. A left ideal has a $\mathbb{Z}$-basis as a free left $\Z$-submodule of rank $n$. The \emph{left inverse} of the ideal $\I$ is defined as $$\I^{-1}=\{x\in\mathbb{Q}[G]\mid xy\in\mathbb{Z}[G],\forall y\in\I\}.$$ An ideal $\I$ is referred to as \emph{left invertible} if its left inverse $\I^{-1}$ is such that $\I^{-1}\I=\mathbb{Z}[G]$. And it can be verified that $\I^{-1}$ is a \emph{left fractional ideal} of $\mathbb{Z}[G]$ which means there exists $t\in\mathbb{Z}$, such that $t\I^{-1}\subseteq\mathbb{Z}[G]$.

In algebraically structured \LWE, an algebraic number field is typically embedded into the real Euclidean space $\R^n$ using canonical embedding as described in \cite{lyubashevsky2010ideal}. However, when dealing with group ring \LWE, we choose the coefficient embedding for simplicity due to the 
relatively complicated non-commutative multiplicative operation.

Since the group ring $\Z[G]$ is a free $\Z$-module, and the elements of $G$ naturally form a $\Z$-basis of $\Z[G]$, we have an embedding $$\varphi:\Z[G]\rightarrow \R^n,\quad \sum\limits_{i=1}^n a_ig_i\mapsto(a_1,a_2,\ldots,a_n)$$ referred to as the \emph{coefficient embedding} of $\Z[G]$, which means the mapping embeds the elements in $\Z[G]$ into $\R^n$ according to their coefficients. Moreover, we can also extend the domain of $\varphi$ to the module tensor product $\Z[G]\otimes_{\Z}\mathbb{R}=\R[G]$, which we denoted by $\bar{\varphi}$. The extension $\bar{\varphi}:\R[G]\rightarrow \mathbb{R}^n$ is bijective. For simplicity, in this paper, we also use the notation ${\varphi}$ instead of $\bar{\varphi}$.

Under the coefficient embedding $\varphi$, any element $\sum_{i=1}^n a_ig_i\in \R[G]$ can be represented uniquely by a vector in $\mathbb{R}^n$, where $n$ is the rank of $\R[G]$ as an $\R$-module. This allows us to define \emph{norm} on elements of $\R[G]$ by taking the corresponding norm of the vector representation. For an element $\mathfrak{h}\in\R[G]$, its norm is defined as  $$\|\mathfrak{h}\|=\left\|\sum_{i=1}^n a_ig_i\right\|:=\|(a_1,a_2,\ldots,a_n)\|,$$ where $\|\cdot\|$ on the right-hand side is any norm defined on vectors in $\R^n$.

Furthermore, it can be easily verified that the coefficient embedding $\varphi$ maps any (fractional) ideal of the ring to a full-rank discrete additive subgroup in $\R^n$, which is consequently a lattice. Here, $n$ is both the rank of the group ring and the lattice. Such lattices induced by a fractional ideal are commonly referred to as  \emph{ideal lattices}. Every (left, right, two-sided/fractional) ideal $\mathcal{I}$ of $\Z[G]$ can be viewed as a lattice in $\R^n$ under coefficient embedding. Moreover, if $\{u_1,u_2,\ldots,u_n\}$ forms a $\Z$-basis of the ideal $\I$, then $\{\varphi(u_1),\varphi(u_2),\ldots,\varphi(u_n)\}\subset\R^n$ forms a basis for the lattice induced by $\mathcal{I}$. Therefore, we can always identify $\mathcal{I}$ as a lattice in $\R^n$. 

Through the study of ideal lattice $\I$ of a group ring, it is crucial to consider the relationship between dual ideal lattice $\I^{\vee}$ and inverse ideal $\I^{-1}$. In the following lemma, we primarily consider the two types of groups defined by \eqref{typei} and \eqref{typeii}. We show that the dual ideal lattice and the inverse ideal lattice of the same invertible ideal lattice are equivalent, up to a specific permutation of the coordinates.
The following lemma generalizes the result of Lemma 3 of \cite{cheng2022lwe} to Type I and Type II. The proof is attached in the appendix.

\begin{lemma}\label{dual ideal}
       For any invertible (right) ideal $\I$ of $\Z[G]$ where $G$ is any group of Type I and Type II, let $\I^{-1}$ be the left inverse of $\I$. Then the dual lattice $\I^{\vee}$ of $\I$ (under coefficient embedding) and $\I^{-1}$ are the same by rearranging the order of the coordinates.
\end{lemma}

In the following lemma, we establish the connection between $\ell_2$ norm of the elements in $\Z[G]$ as defined above, and the matrix norm $\|\cdot\|_{\Mat}$. 

\begin{lemma}\label{matrixnorm}
   Let $\mathbb{R}[G]$ be a group ring where $G$ is a finite group of order $n$. If $\mathfrak{h}\in\mathbb{R}[G]$ is sampled from $D_{r}$ (which means every coefficient of $\mathfrak{h}$ is sampled independently from one-dimensional Gaussian $D_{r}$), then the matrix norm of $\mathfrak{h}$ is at most $nr$ except with negligible probability.
\end{lemma}

\begin{proof}
       We bound the matrix norm of $\mathfrak{h}$ by analyzing its relationship with $\ell_{2}$ norm. Since all group elements of $G$ inherently form an $\R$-basis of $\R[G]$, we obtain a transformation matrix $\M(\mathfrak{h})$ which represents the left multiplication action determined by $\mathfrak{h}$ over such a basis. For any element $\tau\in\R[G]$ with an $\ell_2$-norm equal to 1, we denote its coefficient vector as $\vect$, i.e., $\|\vect\|_{\ell_2}=1.$ By considering the $\ell_2$-norm of $\M(\mathfrak{h})\cdot \vect\in\R^n$, we have $$\|\M(\mathfrak{h})\cdot \vect\|_{\ell_2}\leq \sqrt{n}\cdot\|\mathfrak{h}\|_{\ell_2}\|\vect\|_{\ell_2}$$ after applying the Cauchy-Schwarz inequality. From Lemma \ref{boundgaussian}, the $\ell_{2}$ norm of $\mathfrak{\mathfrak{h}}$ is less than $r\sqrt{n}$ except with probability exponentially closed to 0. Therefore, we obtain $$\|\M(\mathfrak{h})\cdot \vect\|_{\ell_2}\leq nr,$$ which means $\|\mathfrak{\mathfrak{h}}\|_{\Mat}\leq nr$ except with negligible probability.
\end{proof}

\begin{remark}\label{tightbound}
       When applying the same process as described in Lemma 10 of \cite{cheng2022lwe}, we can obtain a tighter bound compared to the one given by Lemma \ref{matrixnorm} in this paper when restricting $G$ to a specific group. It should be pointed out that for group $G=\Z_m\ltimes\Z_n$ of Type I, the matrix norm of the elements in $\R[G]$ which are sampled from $D_r$ can be bounded by $\omega(\sqrt{\log|G|})\cdot\sqrt{|G|}\cdot r$, which is asymptotically smaller than $|G|\cdot r$ obtained from Lemma \ref{matrixnorm}.  However, in this paper, we also use the more general bound as mentioned in Lemma \ref{matrixnorm} for generality.     
\end{remark}

\subsection{Lattice problems}

We introduce some important and useful problems which are believed to be computionally hard. They are commonly used to characterize the hardness of \LWE variants. These problems include \emph{Shortest Vector Problem} (\SVP), \emph{Shortest Independent Vectors Problem} (\SIVP), \emph{Closest Vector Problem} (\CVP), 
and \emph{Bounded Distance Decoding} (\BDD). The following definitions follow the standard formulations used in \cite{peikert2017pseudorandomness}.

\begin{definition}[\SVP and \SIVP]
   Let $\LL$ be an $n$-dimensional lattice and let $\gamma=\gamma(n)\geq 1$. The $\SVP_{\gamma}$ problem in the given norm is: find some nonzero vector $\vecv\in\LL$ such that $\|\vecv\|\leq\gamma\cdot\lambda_{1}(\LL)$. The $\SIVP_{\gamma}$ problem is to find $n$ linearly independent vectors in $\LL$ whose norms are all no more than $\gamma\cdot\lambda_{n}(\LL)$.       
\end{definition}

\begin{definition}[\CVP]
       Let $\LL$ be an $n$-dimensional lattice and let $\gamma\geq 1$. The $\CVP_{\gamma}$ problem in the given norm is: given a target vector $\vect\in\R^{n}$ (which may not be a lattice vector), find some vector $\vecv\in \LL$ such that $\|\vecv-\vect\|\leq\gamma\cdot\|\vecv^\prime-\vect\|$ for any lattice vector $\vecv^\prime$.
\end{definition}

The following problem is a variant version of $\CVP$ where the distance between the target vector and the lattice is bounded.

\begin{definition}[\BDD]
       Let $\LL\subset\R^n$ be a lattice, and let $d<\lambda_{1}(\LL)/2$. The $\BDD_{\LL,d}$ problem in the given norm is: given $\LL$ and $\vecy$ of the form $\vecy=\vecx+\vece$ for some $\vecx\in\LL$ and $\vece\in\R^n$ with $\|\vece\|\leq d$, find $\vecx$.       
\end{definition}

Another problem called \emph{Gaussian Decoding Problem} (\GDP), is essentially \BDD when the offset is sampled from a Gaussian distribution.

\begin{definition}[\GDP]
       For a lattice $\LL\subset\R^n$ and a Gaussian parameter $g>0$. The $\GDP_{\LL,g}$ problem is: given a coset $\vece+\LL$ where $\vece$ is sampled from $D_{g}$, find $\vece$. 
\end{definition}

The following lemma states that if the bound $d$ in \BDD is significantly smaller
compared to the length of the shortest nonzero vector in a lattice $\LL$, then there exists an efficient (within time $\poly(n)$, where $n$ is the dimension of the lattice) algorithm, that can solve $\BDD_{\LL,d}$.

\begin{lemma}[\cite{lenstra1982factoring,DBLP:journals/combinatorica/Babai86}, Babai's Algorithm]
     There is an polynomial-time algorithm that solves $\BDD_{\LL,d}$ for $d=2^{-n/2}\lambda_1(\LL)$.  
\end{lemma}

The following problem,  which is called \emph{Discrete Gaussian Sampling} (\DGS) problem, is needed in the process of reduction mentioned in Section \ref{sec:3} and Section \ref{sec:4}.

\begin{definition}[\DGS]
       Let $R$ be a ring. $R\mbox{-}\DGS_{\gamma}$ asks that given an (invertible) ideal $\I$ of $R$ and a real number $s\geq\gamma=\gamma(\I)$, produce samples from the distribution $D_{\I,s}$.     
\end{definition}

In the following text, we only consider above problems in the context of (invertible) ideal lattices in a certain group ring.

\subsection{Natural inclusion mapping}
When establishing reductions from hard problems in ideal lattices, it is common to consider the relationship between different (ideal) lattices. 

Let $q$ be a positive integer and $\LL$ be a lattice. Denote $\LL_q$ as the quotient $\LL/q\LL$. For any lattices $\LL^\prime\subseteq\LL$, the \emph{natural inclusion map} is defined as $\varphi: \LL_q^\prime\rightarrow\LL_q$ which maps $x+q\LL^{\prime}$ to $x+q\LL$.\footnote{It can be verified that the definition of natural inclusion map is well-defined under the condition that $\LL^\prime\subseteq\LL$.} The natural inclusion map can be viewed as a composition of a natural homomorphism $\LL/q\LL'\rightarrow (\LL/q\LL')/(q\LL/q\LL^{\prime})=\LL/q\LL$ and an inclusion map $\LL^{\prime}/q\LL^{\prime}\rightarrow\LL/q\LL'$.The following lemma is an important result from \cite{peikert2019algebraically}. It describes under what condition the natural inclusion map $\varphi$ is a bijection. 

\begin{lemma}[\cite{peikert2019algebraically}, Lemma 2.13]\label{bijection}
       Let $\LL^{\prime}\subseteq \LL$ be $n$-dimensional lattices and $q$ be a positive integer. Then the natural inclusion map $\varphi:\LL_q^{\prime}\rightarrow\LL_q$ is a bijection if and only if $q$ is coprime with the index $|\LL/\LL^{\prime}|;$ In this case, $\varphi$ is efficiently computable and invertible given an arbitrary basis of $\LL^{\prime}$ relative to a basis of $\LL$. 
\end{lemma}

\begin{remark}
We can also apply Lemma \ref{bijection} to dual lattice $\LL^{\vee}$. As $|\LL/\LL^{\prime}|=|(\LL^{\prime})^{\vee}/\LL^{\vee}|$, we can obtain that $$\psi:(\LL_q^{\prime})^{\vee}\rightarrow\LL_q^{\vee},\quad \psi(x+q(\LL^{\prime})^{\vee})=x+q\LL^{\vee}$$ is also a bijection with the same condition as specified Lemma \ref{bijection}.     
\end{remark}

\subsection{Group ring \LWE}
To prevent the potential attacks that exploit one-dimensional (irreducible) representations of the group $G$, one can use the quotient ring of $\Z[G]$ modulo the sum of ideals associated with its corresponding information-leaking representations. According to the Artin-Wedderburn theorem, we can uniquely decompose the group algebra $\C[G]$ into the direct sum of simple ideals (also simple left $\C[G]$-modules). Moreover, these ideals are isomorphic to matrix rings over $\C$. Precisely, $$\C[G]\cong M_{n_1}(\C)\oplus M_{n_2}(\C)\oplus\cdots\oplus M_{n_r}(\C),$$ where $M_{n_i}(\C)\cong \C^{n_i\times n_i}$ denotes the ring of all $n_i\times n_i$ matrices over $\C$. For each $i$, $M_{n_i}(\C)$ is a simple ideal corresponding to an $n_i$-dimensional irreducible representation of $G$ and $r$ equals the number of non-equivalent irreducible representations of $G$. 

For instance, we consider group $G_1 = \Z_{m}\ltimes\Z_n$ defined by \eqref{typei}. According to Lemma \ref{lemtypei}, the group algebra $\C[G_1]$ can be decomposed as follows: \begin{align*}
       \C[G_1]\cong
       \begin{cases}
       \bigoplus_{i=1}^{m}\C\oplus\bigoplus_{j=1}^{rt}\C^{2\times 2},\quad &\text{if } n \text{ is even};\\
       \bigoplus_{i=1}^{2m}\C\oplus\bigoplus_{j=1}^{r(t-1)}\C^{2\times 2},\quad &\text{if } n \text{ is odd}.
       \end{cases}
       \end{align*}
       
       From representation theory, each direct summand $\B$ in the Artin-Wedderburn decomposition of $\C[G]$ above is a minimal principal ideal of $\C[G]$ with a generator referred to as \emph{central primitive idemptotent}. \cite{endelman2008primitive} provides a method to determine all central primitive idempotents of the group algebra $\C[G]$ where $G$ is the semi-direct product of two finite Abelian groups. Consequently, we can easily compute the simple ideal of $\C[G]$ corresponding to one-dimensional representations of $G$ when $G$ is either of Type I or Type II.
       
       \begin{example}
          According to the Artin-Wedderburn theorem, we can decompose $\C[D_{2n}]$ into the direct sum of simple ideals. When $n$ is even, we have $$\C[D_{2n}]\cong\bigoplus_{i=0}^{3}\C\oplus\bigoplus_{i=4}^{(n+4)/2}\C^{2\times 2}.$$ By \cite{endelman2008primitive}, we can obtain all the central primitive idempotents corresponding to one-dimensional representation of $D_{2n}$ as follows: \begin{align*}
              &(1+t+t^2+\cdots+t^{n-1})(1\pm s),\\
       &(1-t+t^2+\cdots-t^{n-1})(1\pm s).
          \end{align*}
       Each of these four idempotents generates a simple ideal and each of these ideal divides $\langle t^{n/2}+1\rangle$ when $4\mid n$. Thus we can assert that $\C[D_{2n}]/\langle t^{n/2}+1\rangle$ eliminate the potential attack making use of the 1-dimensional representations of $D_{2n}$.

       When $n$ is odd, there are two one-dimensional irreducible representations for $D_{2n}$, and the corresponding central primitive idempotents are given by $$(1+t+t^2+\cdots+t^{n-1})(1\pm s).$$ In this case, we can select $\C[D_{2n}]/\langle 1+t+t^2+\cdots+t^{n-1}\rangle$ for our purpose.
       \end{example}
       
       Here we describe the \LWE problem over group rings. We use the following conventional notations from \cite{lyubashevsky2010ideal,cheng2022lwe}. Let $G$ be a finite group and let $R$ be the group ring $\Z[G]$ itself or one of its quotient rings. For an integer modulus $q\geq 2$, let $R_q$ denote the quotient ring $R/qR$. Likewise, for any (fractional) ideal $\I$ of a ring $R$, let $\I_q$ denote $\I/q\I$. We also denote $R_{\R}$ as the tensor product $R\otimes_{\Z}\R$ and let $\T:=R_{\R}/R$.

\begin{definition}[Group ring-\LWE distribution]
       For a secret element $s\in R_{q}$ and an error distribution $\psi$ over $R_{\R}$, a sample chosen from the $R\mbox{-}\LWE$ distribution $A_{s,\psi}$ over $R_{q}\times\T$ is generated by sampling $a\in R_q$ uniformly and $e\leftarrow\psi$, and then outputting $(a, (s\cdot a)/q+e\bmod{R})$.       
 \end{definition}
 
 \begin{definition}[\cite{cheng2022lwe}, Search Group ring-\LWE]
        Let $\Psi$ be a family of distributions over $R_{\R}\times\T$. The \emph{search} version Group ring-\LWE problem asks to find the secret element $s\in R_q$, given arbitrarily many independent samples $(a_i,b_i)\in R_q\times\mathbb{T}$ chosen from the $R\mbox{-}\LWE$ distribution $A_{s,\psi}$ for some arbitrary $s\in R_q$, where $\psi$ is a distribution of $\Psi$. We denote such search problem by $R\mbox{-}\LWE_{q,\Psi}$.
 \end{definition}
 
 \begin{definition}[\cite{cheng2022lwe}, Average-case decision Group ring-$\LWE$]
       Let $\Upsilon$ be a distribution over a family of distributions, each over $R_{\R}$. The \emph{(average-case) decision} version Group ring-$\LWE$ problem asks to distinguish between polynomially many samples from $R\mbox{-}\LWE$ distribution  $A_{s,\psi}$ for a uniformly random $(s,\psi)\leftarrow R_q\times\Upsilon$, and the same number of uniformly random and independent samples from $R_q\times\mathbb{T}$ with non-negligible advantage. We denote such decision problem by $R\mbox{-}\DLWE_{q,\Upsilon}$.
 \end{definition}

\subsection{Oracle hidden center problem}
In the process of reduction to $R\mbox{-}\DLWE$, it is necessary to analyze the properties of the given $R\mbox{-}\DLWE$ oracle. In \cite{peikert2017pseudorandomness}, Peikert et al. provided a direct reduction from worst-case lattice problem to decision Ring-\LWE. This reduction possesses tighter parameters and more compact error rates compared to those in \cite{lyubashevsky2010ideal}, where the authors established hardness by additionally reducing search Ring-\LWE to decision Ring-\LWE. The technique introduced in \cite{peikert2017pseudorandomness} exploits the properties of a suitable decision \LWE oracle, which is abstracted as an ``oracle with a hidden center". Such an oracle enables dealing with \BDD problem in ideal lattices as needed in the proof of the reduction, thus facilitating the hardness proof. 
 
\begin{definition}[\cite{peikert2017pseudorandomness}, Oracle Hidden Center Problem (OHCP)]\label{OHCP}
       For any $\varepsilon,\delta\in [0,1) $ and $\beta\geq 1$, the $(\varepsilon,\delta,\beta)$-OHCP is an approximate search problem defined as below. An instance consists of a scale parameter $d>0$ and randomized oracle $\OO:\R^{k}\times \R^{\geq 0}\rightarrow \{0,1\}$ which satisfies for an input $(\vecz,t)$ for $\|\vecz-\vecz^{\ast}\|\leq \beta d$, $$\Pr(\OO(\vecz,t)=1)=p(t+\log\|\vecz-\vecz^{\ast}\|),$$ for some (unknown) ``hidden center" $\vecz^\ast\in\R^k$ with $\delta d\leq\|\vecz^{\ast}\|\leq d$ and some (unknown) function $p$. The goal is to output some $\tilde{\vecz}\in\R^{k}$ such that $\|\tilde{\vecz}-\vecz^{\ast}\|\leq\varepsilon d$.       
 \end{definition}

 In \cite{peikert2017pseudorandomness}, Peikert et al. showed that there is an efficient algorithm to solve OHCP if the oracle of the instance satisfies certain conditions. Specifically, it states as follows.

 \begin{proposition}[\cite{peikert2017pseudorandomness}, Proposition 4.4]\label{ohcp}
       There is a $\poly(\kappa,k)$-time algorithm that takes as input a confidence parameter $\kappa\geq 20\log(k+1)$ with the scale parameter $d>0$ and solves $(\exp(-\kappa),\exp(-\kappa),1+1/\kappa)$-OHCP in dimension $k$ with accept probablity greater than $1-\exp(-\kappa)$, provided that the oracle $\OO$ corresponding to the OHCP instance satisfies the following conditions. For some $p_{\infty}\in [0,1]$ and $s^{\ast}\geq 0$,
       \begin{enumerate}
              \item $p(s^{\ast})-p_{\infty}\geq 1/\kappa$;
              \item $|p(s)-p_\infty|\leq 2\exp(-s/\kappa)$ for any $s$;
              \item $p(s)$ is $\kappa$-Lipschitz in $s$, i.e., $|p(s_1)-p(s_2)|\leq\kappa|s_1-s_2|$ for all $s_1,s_2$,
       \end{enumerate}
       where $p(s)$ is the acceptance probability of  $\OO$ on input $(\veczero,s).$
\end{proposition}

\section{The hardness of search $\GRLWE$}
\label{sec:3}

For $G_1=\Z_m\ltimes\Z_n=\langle s\rangle\ltimes\langle t\rangle$ of Type I with $m$ an even integer and $4\mid n$, we choose the group ring $R^{(1)}=\Z[G_1]/\langle t^{n/2}+1\rangle$. The reason for selecting this group is the same as mentioned in \cite{cheng2022lwe}. This ring does not have direct summands corresponding to one-dimensional representations, thereby ensuring its resistance against the aforementioned potential attacks.  Similarly, for $G_2=\Z_{p^k}^{\ast}\ltimes\Z_{p^k}$ of Type II where we denote the generator of $\Z_{p^k}$ as $g$, we select the group ring $R^{(2)}=\Z[G_{2}]/\langle 1+g+g^2+\cdots+g^{p^k-1}\rangle.$ 

\subsection{Main result}

We claim that for rings $R^{(1)}$ and $R^{(2)}$ (or more generally, the group ring with equivalent dual ideal and inverse ideal up to a certain permutation under coefficient embedding), the hardness of search \LWE problems over them is based on the hardness of finding short vectors in related ideal lattices, similar to the reduction in \cite{lyubashevsky2010ideal}. To be specific, we state the results as 
follows. From this section, we denote $\omega(f(n))$ as some fixed function that grows asymptotically faster than $f(n)$. Additionally, we define the family  $\Psi_{\leq\alpha}$ for a positive real $\alpha$ as the set of all elliptical Gaussian distributions $D_{\vecr}$ with each coordinate $r_i\leq\alpha$.

\begin{theorem}[Main Result]\label{searchreduct}
 Let $R=\Z[G]$ be a group ring where $G$ is of Type I or Type II with $n$ elements. Let $\alpha=\alpha(n)>0$, and let $q=q(n)\geq 2$ be such that $\alpha q\geq 2n$. For some negligible $\varepsilon = \varepsilon(n)$, there is a probabilistic polynomial-time quantum reduction from $R\mbox{-}\DGS_{\gamma}$ (and hence $\SIVP$ with approximate factor $\tilde{O}(n^{3/2}/\alpha)$) to (search) $R\mbox{-}\LWE_{q,\Psi_{\leq\alpha}}$, where \begin{align}\label{cut-off2}
       \gamma=\max\{\eta_{\varepsilon}(\I)\cdot(\sqrt{2}n/\alpha),2\sqrt{n}/\lambda_{1}(\I^{\vee})\}.
 \end{align}     
\end{theorem}

Recall that the $R$-\DGS problem asks to sample from discrete Gaussian distribution over the (invertible) ideal $\I$ efficiently. In \cite{regev2009lattices}, the author presented a direct reduction from standard lattice problems to the \DGS problem. Combining this result, Theorem \ref{searchreduct} accomplishes the reduction from lattice problems to search $\GRLWE$. To be precise, based on Lemma \ref{smooth} and Claim 2.13 in \cite{regev2009lattices}, we know that $1/\lambda_{1}(\I^{\vee})\leq\eta_{\varepsilon}(\I)\leq\lambda_{n}(\I)\cdot\omega(\sqrt{\log n})$ for any ideal lattice $\I$ in $\R^{n}$. By Lemma \ref{boundgaussian},  we obtain the $\ell_2$-norm of samples from $D_{\I,\gamma}$ is at most $\gamma\sqrt{n}$ except with negligible probability. Consequently, we can use the outputs of $\DGS_{\gamma}$ as a solution of $\SIVP$ with approximate factor $\tilde{O}(n^{3/2}/\alpha)$ when $\alpha$ is restricted to be no greater than $\sqrt{n}$.  This problem is believed to be a computationally hard problem, and the restriction is always satisfied to make the problem information-theoretically solvable. 

\begin{remark}
       For the sake of completeness, we prove the results for the group ring $\Z[G_1]$ and $\Z[G_2]$ rather than the quotient groups $R^{(1)}$ and $R^{(2)}$. In fact, using the same procedure and fundamental homomorphism theorem, we can also prove the same result for $R^{(1)}$ and $R^{(2)}$. 
       It is also worth noting that by applying the results mentioned in Remark \ref{tightbound}, the approximate factor can be further optimized to $\tilde{O}(n/\alpha)$ on these two particular groups.
\end{remark}

\begin{proof}[Proof of Theorem \ref{searchreduct}]
       According to Lemma 3.2 in \cite{regev2009lattices}, one can sample  efficiently from $D_{\I,r}$ for sufficiently large $r$, say $r>2^{2n}\lambda_n(\I)$. In this case, polynomially many samples from $D_{\mathcal{I},r}$ can be generated typically by the following steps. First, generate sample $y$ from (continuous) Gaussian distribution $D_{r}$ using the standard method, then output $y-(y\bmod \I)\in\mathcal{I}$. Next, we can repeatedly apply the reduction specified by Lemma \ref{iterative} (but still polynomial times) mentioned later. This allows us to sample from discrete Gaussian distribution with narrower and narrower parameters. Since $\alpha q\geq 2n$, the iterative steps enable us to sample from $D_{\I,r/2}$ given polynomially many samples from $D_{\I,r}$. The repeating iteration continues until the Gaussian parameter reaches the desired value $s\geq \gamma$. Finally, the procedure ends up with one (or more) sample from $D_{\I,s}$.    
\end{proof}

\subsection{The iterative step}

The proof of the iterative step basically follows the procedure outlined in \cite{regev2009lattices} and \cite{lyubashevsky2010ideal}. The reduction uses repeated \emph{iterative steps} to achieve the goal. The iterative step states that when the initial Gaussian parameter $r$ is sufficiently larger than the smoothing parameter, we can sample efficiently from another discrete Gaussian distribution with narrower parameters, say $r/2$.

\begin{lemma}[The iterative step]\label{iterative}
       Let $R=\Z[G]$ be a group ring where $G$ is of Type I or Type II with $n$ elements. Let $\alpha>0$ and let $q>2$ be an integer. There exists an efficient quantum algorithm that, given an invertible ideal $\I$ in $R$ satisfying that $\det(\I)$ is coprime with $q$, a real number $r\geq \sqrt{2}q\cdot \eta_{\varepsilon}(\I)$ for some negligible $\varepsilon=\varepsilon(n)>0$ such that $r':=rn/\alpha q>2\sqrt{n}/\lambda_{1}(\I^{\vee})$, an oracle to $R\mbox{-}\LWE_{q,\Psi_{\leq\alpha}}$, and a list of samples from the discrete Gaussian distribution $D_{\I,r}$ (as many as required by the $R\mbox{-}\LWE$ oracle), outputs an independent sample from $D_{\I,r'}$. 
\end{lemma}
The iterative step stated above can be partitioned into two parts as in \cite{regev2009lattices}.
The first part of the iteration is classical, which shows that given a search $R$-\LWE oracle, we can solve \GDP on $\I^{\vee}$ making use of the given discrete Gaussian samples. The proof of this part (Lemma \ref{BDDtoLWE}) is presented in Section 3.3.

\begin{lemma}[Reduction from \GDP to \LWE]\label{BDDtoLWE}
       Let $\varepsilon=\varepsilon(n)$ be some negligible function, let $q> 2$ be an integer, and let $\alpha\in(0,1)$ be a real number. Let $R=\Z[G]$ be a group ring where $G$ is of Type I or Type II with $n$ elements, and let $\I$ be an invertible ideal in $R$ satisfying that $\det(\I)$ is coprime with $q$. Given an oracle for discrete Gaussian distribution $D_{\I,r}$, where $r\geq\sqrt{2}q\cdot \eta_{\varepsilon}(\I)$, there is a probabilistic polynomial-time (classical) reduction from $\GDP_{\I^{\vee},d/n}$ to $R\mbox{-}\LWE_{q,\Psi_{\leq \alpha}}$, where $d=\alpha q/(\sqrt{2}r)$.
      \end{lemma}

      The second part of the iteration was initially proposed by \cite{regev2009lattices} and later improved by \cite{lyubashevsky2010ideal}. It is worth noting that this part is the only quantum component of the whole reduction. The lemma in the following states that we can use a \GDP oracle to sample polynomially many lattice vectors with a narrower width. By employing Lemma \ref{matrixnorm}, we can essentially derive a similar lemma as in \cite{cheng2022lwe}.

\begin{lemma}[\cite{regev2009lattices}, Lemma 3.14]\label{quantum}
       There is an efficient quantum algoithm that, given any $n$-dimensional lattice $\LL$, a number $d'<\lambda_1(\LL^{\vee})/2$ (where $\lambda_1$ is in $\ell_2$-norm), and an oracle that solves $\GDP_{\LL^{\vee},d'/\sqrt{2n}}$, outputs a sample from $D_{\LL^{\vee},\sqrt{n}/\sqrt{2}d'}$.
\end{lemma}

Combining the results of Lemma \ref{BDDtoLWE} and Lemma \ref{quantum}. We can prove Lemma \ref{iterative} as follows.
\begin{proof}[Proof of Lemma \ref{iterative}]
       By Lemma \ref{BDDtoLWE}, given samples from $D_{\I,r}$ and oracle for search $R\mbox{-}\LWE_{q,\Psi_{\leq\alpha}}$, we can solve $\GDP_{\I^{\vee},d/n}$ problem with parameter $d=\alpha q/\sqrt{2}r$. By Lemma \ref{quantum} and setting $d/n=d'/\sqrt{2n}$, we have $$d'=\sqrt{2}d/\sqrt{n}=\sqrt{n}/r'<\lambda_1(\I^{\vee})/2,$$ where the last equality comes from the condition mentioned in Lemma \ref{BDDtoLWE}. Thus we obtain samples from $D_{\I,r'}$. 
\end{proof}

\subsection{The \GDP to Search $\GRLWE$ reduction}
In this section, our goal is to prove Lemma \ref{BDDtoLWE}, which means providing a reduction from \GDP problem in ideal lattices to $\GRLWE$. In \cite{regev2009lattices}, it has been proven that to solve \BDD in some lattice $\LL$, it is sufficient to find a close vector modulo $q\LL$. We present a special case of Lemma 3.5 in \cite{regev2009lattices} as follows and the proof follows essentially the same approach.

\begin{lemma}[\cite{regev2009lattices}, Lemma 3.5]\label{BDD-QBDD}
       For any $q\geq 2$, there is a deterministic polynomial-time reduction from $\BDD_{\I,d}$ (in matrix norm) to $q$-$\BDD_{\I,d}$ (in the same norm).
       \end{lemma}

 When dealing with coefficient embedding, it is common to consider sampling vectors from an (ideal) lattice according to a discrete Gaussian distribution. In the work of \cite{regev2009lattices}, the spherical Gaussian distribution was considered, and later in \cite{lyubashevsky2010ideal}, this distribution was generalized to non-spherical Gaussians. However, in the case of $\GRLWE$, especially with coefficient embedding, a more generalized Gaussian is required. This distribution should have an arbitrary definite positive matrix as its covariance matrix, rather than a diagonal matrix. According to Lemma \ref{BDD-QBDD}, it suffices to give a reduction from $q$-\GDP to $\GRLWE$. Before showing this reduction, it is necessary to introduce some lemmas concerning smoothing parameters. The smoothing parameter characterizes how a discrete Gaussian over a certain lattice behaves similarly to a continuous Gaussian. As a generalization of Claim 3.9 of \cite{regev2009lattices}, we provide the following lemma and corollary, illustrating that when a discrete Gaussian with the smoothness condition is added to a continuous Gaussian, it ``acts like" a continuous Gaussian with the same covariance matrix, up to a negligible statistical distance under certain ``smoothness condition''. 

\begin{lemma}[\cite{cheng2022lwe}, Lemma 5]\label{smoothness}
       Let $\LL$ be a lattice of $\R^n$. Let $\matA,\matB$ be two fixed non-singular matrices. Assume that  smoothness condition $$\sum\limits_{\vecy\in\LL^{\vee}\backslash\{0\}}\exp(-\pi\vecy^{t}(\matA^{-t}\matA^{-1}+\dfrac{1}{s^2}\matB^{t}\matB)^{-1}\vecy)\leq\varepsilon$$ holds for some negligible $\varepsilon$. Let $\vecv$ be distributed as discrete Gaussian $D_{\LL+\vecu, \matA}$ for arbitrary $\vecu\in\R^{n}$ and let $\vece^{\prime}\in \R^n$ be  distributed as $n$-dimensional spherical (continuous) Gaussian $D_{\sigma}$. Then the distribution of $\matB\cdot \vecv+\vece^{\prime}$ is within statistical distance $4\varepsilon$ of Gaussian distribution $D_{\matC}$, where $\matC=\frac{1}{2\pi}\matB\matA\matA^{t}\matB^{t}+\frac{\sigma^2}{2\pi}\matI_n$.
\end{lemma}

By applying this lemma to a fractional ideal in $\R[G]$ (where the elements are under coefficient embedding), we obtain the following corollary, which generalizes Corollary 1 of \cite{cheng2022lwe}.

\begin{corollary}\label{groupsmoo}
       Let $G$ be a finite group of order $n$, and let $\I$ be an arbitrary fractional ideal in the group ring $\R[G]$. Let $\mathfrak{h}$ be some element from $\R[G]$ and $\alpha=\|\mathfrak{h}\|_{\operatorname{Mat}}$. Let $r,s>0$ be two reals and $t=1/\sqrt{1/r^2+\alpha^2/s^2}$. Assume that the smoothness condition $$\sum\limits_{y\in\I^{\vee}\backslash\{0\}}\exp(-\pi t^2\cdot\|y\|^2)\leq\varepsilon$$ holds for some negligible $\varepsilon=\varepsilon(n)>0$. Let $v$ be ditributed as $D_{\I+u,r}$ for arbitrary $u\in\R[G]$, and let $e$ be sampled from $n$-dimensional Gaussian $D_{s}$. Then the distribution of $\mathfrak{h}\cdot v+e$ belongs to $n$-dimensional family $\Psi_{\leq\sqrt{r^2\alpha^2+s^2}}$ (under some unitary base transformation).
\end{corollary}

\begin{proof}
       Under coefficient embedding $\varphi$, each element in the group ring $\R[G]$ is mapped to a vector in $\R^n$. Let $\M(\mathfrak{h})$ be the matrix representation of $\mathfrak{h}$. Then the element $\mathfrak{h}v+e$ is mapped to $\M(\mathfrak{h})\varphi(v)+\varphi(e)\in\R^n$. By setting $\matA=r\cdot I_n$ and $\matB=\M(\mathfrak{h})$, we can easily obtain that the largest absolute eigenvalues of $\matA^{-1},\matB$ are $1/r$ and $s$, respectively. Thus, the smallest absolute eigenvalue of $(\matA^{-t}\matA^{-1}+{1}/{s^2}\cdot\matB^{t}\matB)^{-1}$ is bounded by $\sqrt{1/(1/r^2+\alpha^2/s^2)}=t$. Consequently, $$\vecy^t(\matA^{-t}\matA^{-1}+\dfrac{1}{s^2}\matB^{t}\matB)^{-1}\vecy\geq t^2\|\vecy\|^2$$ holds for any $\vecy\in\R^n$, which means smoothness condition described in Lemma \ref{smoothness} is satisfied. 
\end{proof}

\begin{remark}\label{onlydepend}
       Through a certain unitary basis transformation, we can obtain an alternative $\Z$-basis of $\Z[G]$. By adjusting the $\Z$-basis properly, the $n$ coefficients of $\matB\cdot\varphi(v)+\varphi(e)$ are distributed according to a Gaussian distribution with a diagonal covariance matrix. In this matrix, the absolute value of each diagonal element is not greater than $\sqrt{r^2\alpha^2+s^2}$.  Based on the aforementioned proof, it can be observed that if all the eigenvalues of $\M(\mathfrak{h})$ are $\lambda_1,\lambda_2,\ldots,\lambda_{n}$, then we can perform a (known) unitary basis transformation to convert $\mathfrak{h}v+e$ to a sample from the diagonal Gaussian distribution $\prod_{i=1}^n D_{\sqrt{r^2\lambda_i^2+s^2}}$. Thus, for group rings where the underlying finite group is of Type I and Type II, we could compute the resulting diagonal Gaussian distribution by combining the results of Lemma \ref{eigen}. 
\end{remark}

To achieve this goal, we introduce the following lemma, which demonstrates how to convert a $q$-\BDD instance to a $\GRLWE$ instance.
As mentioned in \cite{regev2009lattices}, it suffices to use an $\LWE_{q,\Psi_{\leq\alpha}}$ oracle to give a solution efficiently to an instance of $\LWE_{q,\Psi_{\leq\beta}}$ for any $\beta\leq\alpha$, even without knowing the exact value of $\beta$. Therefore, it is unnecessary to compute all the entries of the covariance matrix of $\mathfrak{h}\cdot v+e$. We only need to give an upper bound of the eigenvalues of the covariance matrix. The following lemma plays a crucial role in reductions for both the search version and the decision version of $\GRLWE$. It states that under certain mild conditions, given an \BDD instance, there exists an efficient algorithm that generates an $\GRLWE$ sample.

\begin{lemma}\label{constructlwe}
       Let $R=\Z[G]$ be a group ring, where $G$ is a group of Type I or Type II with $n$ elements. Let $\alpha>0$ be a real, let $q> 2$ be an integer, and let $r>\sqrt{2}q\cdot\eta_{\varepsilon}(\I)$ be a real. Then given an invertible right ideal $\I$ of $R$ with index $|R/\I|=\det(\I)$ coprime with $q$, there exists an efficient algorithm that given an instance from $\BDD_{\I^{-1},d}$ (in matrix norm), where $\I^{-1}$ is the left inverse of $\I$ and samples from $D_{\I,r}$,  outputs samples that have $R$-\LWE distribution $A_{q,\psi}$ (up to negligible statistical distance) for some $\psi\in\Psi_{\leq\alpha}$, where $d=\alpha q/(\sqrt{2}r)$.
\end{lemma}

\begin{proof}
       For an $\BDD_{\I^{-1},d}$ instance $y=x+e$ where $x$ is an element of $\I^{-1}$ and the matrix norm of $e$ is bounded by $d$, we can construct an $R\mbox{-}\LWE$ sample as follows. 
       
       First, sample $z\leftarrow D_{\I,r}$ from the oracle of discrete Gaussian distribution. Since $r$ exceeds the smoothing parameter of $\I$, $z\bmod q\I$ is almost uniformly distributed in $\I_q$ (up to negligible statistical distance). According to Lemma \ref{bijection}, we know the natural inclusion map $$\varphi:\I/q\I\rightarrow R/qR$$ is a bijection. Thus $a:=\varphi(z\bmod q\mathcal{I})=z\bmod{qR}$, which is also uniform in $R_q$. Moreover, we have another natural inclusion map $$\rho:R/qR\rightarrow \I^{-1}/q\I^{-1},$$ where $\I^{-1}$ is the left inverse of $\I$. We construct an element $$b=(y\cdot z)/q+e^{\prime}\bmod R=(x\cdot z)/q+(e\cdot z)/q+e'\bmod R,$$ where $e'$ is an error sampled from continuous Gaussian $D_{\alpha/\sqrt{2}}$. We claim that the pair $(a,b)$ is an $R$-\LWE sample.

       We first consider the element $x\cdot z \bmod qR$. From the property of natural inclusion mapping, there exists a unique $\bar{s}=s+qR\in R/qR$, such that $$\rho(\bar{s})=s\bmod q\I^{-1}=x+q\I^{-1},$$ Note that $$x\cdot z+qR=(x+qR)(z+qR)=\rho^{-1}(x+q\I^{-1})\varphi(z+q\I)=(s+qR)(z+qR)=\bar{s}\cdot \bar{a}.$$ We obtain that $x\cdot z=s\cdot a\bmod qR.$

       It remains to analyze the covariance matrix of $(e\cdot z)/q+e^{\prime}$. Since $\|e\|_{\operatorname{Mat}}\leq \alpha\cdot q/(\sqrt{2}r)$, and $z$ is distributed as $D_{\I,r}$, then it can be verified the smoothness condition holds: $$\sum_{y\in \I^{-1}\backslash\{0\}}\exp(-\pi t^2\|y\|^{2})=\sum_{y\in \I^{-1}\backslash\{0\}}\exp(-\pi\dfrac{r^2}{2q^2}\|y\|^{2})\leq\varepsilon$$ where $t=1/\sqrt{(q/r)^2+q^2/r^2}$ as in Lemma \ref{groupsmoo}, we know $(e\cdot z)/q+e^{\prime}$ is distributed as some $\psi\in\Psi_{\leq\alpha}$.
\end{proof}

From the reduction above, we can obtain an $R\mbox{-}\LWE$ sample given some samples from $D_{\I,r}$ with $r$ exceeding the smoothing parameter of $\I$. By employing the search $R\mbox{-}\LWE$ oracle, we can recover $s\bmod qR$ except with negligible probability, allowing us to use the bijective mapping from $R/qR$ to $\I/q\I$ to address the \BDD problem. Thus, we can prove Lemma \ref{BDDtoLWE}.

\begin{proof}[Proof of 
       Lemma \ref{BDDtoLWE}]
       We can observe that an instance of $\GDP_{\I^{-1},d/n}$ can be inherently regarded as an instance of $\BDD_{\I^{-1},d}$ (in the matrix norm). According to Lemma \ref{BDDtoLWE}, we can convert the instance into an $R\mbox{-}\LWE$ sample $(a,b)$. After inputting $(a,b)$ into the given search $R\mbox{-}\LWE$ oracle, we can get a solution $\bar{s}\in R/qR$. Next, we calculate $\rho(\bar{s})\in \I^{-1}/q\I^{-1}$ which equals $x\bmod q\I^{-1}$. Consequently, we obtain a solution of the $q\mbox{-}\GDP_{\I^{-1},d/n}$ with the instance $y$. Additionally, we have exploited the mild properties of groups of Type I and Type II. Specifically, the dual ideal $\I^{\vee}$ and inverse ideal $\I^{-1}$ (if it exists) of $\Z[G]$ are equivalent up to a (known) permutation, which means $\GDP_{\I^{\vee},d/n}$ and  $\GDP_{\I^{-1},d/n}$ are essentially equivalent. Hence we have proven the lemma.
\end{proof}

\section{Hardness of decision $\GRLWE$}
\label{sec:4}
Having given the reduction from worst-case lattice problem to search \LWE, 
Regev \cite{regev2009lattices} also established a reduction from search \LWE to (average-case) decision \LWE, which provides the basis for hardness in the decisional setting. Similarly, Lyubashevsky et al.  \cite{lyubashevsky2010ideal} also presented such reductions in the cyclotomic ring version. However, these reductions from search Ring-\LWE to decision Ring-\LWE have more restrictions on the underlying rings and result in worse parameters. Later in \cite{peikert2017pseudorandomness}, Peikert et al. showed a direct and tighter reduction from worst-case (ideal) lattice problem to decision Ring-\LWE with more compact error rates.
In this section, we study the hardness of the decisional version of $\GRLWE$, with a proof similar to \cite{peikert2017pseudorandomness}.

\begin{theorem}[Main Result] Let $R=\Z[G]$ be a group ring where $G$ is of Type I or Type II with $n$ elements. Let $\alpha=\alpha(n)>0$, and let $q=q(n)> 2$ be an integer such that $\alpha q\geq 2n$. For some negligible $\varepsilon = \varepsilon(n)$, there is a probabilistic polynomial-time quantum reduction from $R\mbox{-}\DGS_{\gamma}$ (and hence $\SIVP$ with approximate factor $\tilde{O}(n^{3/2}/\alpha)$) to $R\mbox{-}\DLWE_{q,\Psi_{\leq\alpha}}$, where \begin{align}\label{cut-off}
       \gamma=\max\{\eta_{\varepsilon}(\I)\cdot(\sqrt{2}n/\alpha),2\sqrt{n}/\lambda_{1}(\I^{\vee})\}.
 \end{align}        
\end{theorem}

\begin{remark}
       When employing the tighter bound mentioned in Remark \ref{tightbound}, it is possible to improve the parameter $\gamma$ to $$\max\{\eta_{\epsilon}(\I)\cdot(\sqrt{2n}/\alpha)\cdot\omega(\sqrt{\log n}),\sqrt{2n}/\lambda_{1}(\I^{\vee})\},$$ provided that $\alpha q\geq \sqrt{n}\cdot\omega(\sqrt{\log n})$, which shows a reduction from $\SIVP$ problem with approximate factor $\tilde{O}(n/\alpha)$.
\end{remark}

For simplicity, we only provide the reduction for group ring with underlying group of Type I. The one for Type II is essentially the same. The reduction process remains the same through repeated iterative steps, as stated in Section \ref{sec:3}. The only difference is that we have access to a decision $\GRLWE$ oracle instead of a search version one. To begin with, we introduce some notations of a special family of polynomials in $\R[x]$.

\begin{definition}\label{lagrange}
       Let reals $r>0,\iota>0$, an integer $T\geq 1$ and let $\xi$ be a $v$-th primitive root of unity. Let $W_{r,\iota,T}$ be any set\footnote{The specific selection of $W_{r,\iota,T}$ does not affect the analysis in the following.} of polynomials containing for each $i=0,1,\ldots, v-1, k=0,1,\ldots,T$, $r_{k}^{(i)}(x)\in\R[x]$ which denotes a polynomial such that \begin{align}
              &r_{k}^{(i)}(\xi^{\ell})=r, \quad\forall \ell\neq i, v-i,\label{W1}\\ &r_{k}^{(i)}(\xi^{i})=r_{k}^{(i)}(\xi^{v-i})=r(1+\iota)^{k}.\label{W2}
\end{align}
 \end{definition}

 \begin{example}
       Let's consider the group ring with an underlying group of Type I: $G=\Z_u\ltimes\Z_v=\langle s\rangle \ltimes\langle t \rangle$, where $t$ has order $v$. Any element $\mathfrak{h}\in \R[G]$ of the form $$\mathfrak{h}=\sum_{i=0}^{v-1}f_{i}t^i,\quad f_{i}\in\R,$$ can also be regarded as a polynomial in $\R[x]$ with degree no greater than $v-1$ by replacing $t$ with the indeterminate $x$. It can be easily verified that there is at least one element in $\R[G]$ satisfying the evaluation \eqref{W1} and \eqref{W2} in Definition \ref{lagrange}, which means $W_{r,\iota,T}$ is well-defined. In fact, we can choose the elements by 
       the same method as Definition 11 in \cite{cheng2022lwe}.  through Lagrange interpolation.  
 \end{example}

 The lemma in the following, referred to as iterative steps, combines the results of Lemma \ref{BDDtoLWE2} and Lemma \ref{quantum}.

 \begin{lemma}[The iterative step]\label{iterative2}
       Let $R=\Z[G]$ be a group ring where $G$ is of Type I or Type II with $n$ elements. There exists an efficient quantum algorithm that given an oracle that can solve $R\mbox{-}\DLWE_{q,\Psi_{\leq\alpha}}$ on input a number $\alpha\in (0,1)$ and an integer $q\geq 2$, an (invertible) ideal $\I\in \Z[G]$ with $\det(\I)$ coprime with $q$, a real number $r\geq \sqrt{2}q\cdot\eta(\I)$ such that $r^{\prime}:= r\cdot n/(\alpha q)>2\sqrt{n}/\lambda_{1}(\I^{\vee})$, polynomially many samples from discrete Gaussian distribution $D_{\I,\matA}$ where $\matA$ is the matrix representation for each $r_{k}^{(i)}\in W_{r,\iota, T}$, and a vector $\vecr^{\prime}\in\R^n$ with each coordinates $r_{i}^{\prime}>r$, outputs an independent sample from $D_{\I,\vecr^{\prime}}$.         
 \end{lemma}

 The following lemma originates from Lemma 6.6 of \cite{peikert2017pseudorandomness} and is a slightly stronger version of Lemma 6 of \cite{cheng2022lwe}. 

 \begin{lemma}\label{largesmooth}
       Let $R=\Z[G]$ be a group ring of order $n$ with $G=\Z_u\ltimes \Z_v=\langle s\rangle\ltimes\langle t\rangle$ of Type I, where $t$ is an element of order $v$. Let $\xi$ be a $v$-th primitive root of unity. Let $r(x)\in\mathbb{R}[x]$ be a  polynomial with degree no greater than $v-1$, and let $$c=\left(\prod_{i=0}^{v-1}(1/\sqrt{v})r(\xi^i)\right)^{1/v}\geq 1.$$ Then the matrix determined by $r(t)$, which is denoted by $\matA$, satisfies the smoothness condition: $$\sum_{y\in R^{\vee}\backslash\{0\}}\exp(-\pi\cdot\vecy^t\matA\matA^t\vecy)\leq\varepsilon,$$ where $\varepsilon=\exp(-c^2v)$.
 \end{lemma}

 \begin{proof}
       Since $r(t)$ is invertible, the matrix $\matA$ is also invertible. According to Lemma \ref{smooth}, we have \begin{align}\label{eqsmooth}
              \eta_{\varepsilon}(\matA^{-1}R)\leq c\sqrt{v}/\lambda_{1}((\matA^{-1}R)^{\vee}),\end{align} where the dual lattice satisfies $(\matA^{-1}R)^{\vee}=\matA^{T}\mathbb{Z}^{n}$ under coefficient embedding. Note that the lattice $\matA^{T}\mathbb{Z}^{n}$ can be viewed as the concatenation of a series of $r(t)\Z[t]/\langle t^v-1\rangle$ and $r(t^{-1})\Z[t]/\langle t^v-1\rangle$ (each number of which is dependent on the order of $s$). For any polynomial $f=\sum_{i=0}^{v-1}f_ix^i\in\mathbb{R}[x]$, the norm of $f$ (under coefficient embedding) is given by $$\|f\|_{2}^{2}=\sum_{i=0}^{v-1}f_{i}^2=\sum_{i=0}^{v-1}|f(\xi^{i})/\sqrt{v}|^2.$$ For any polynomial $g(x)\in\Z[x]$ of degree less than $v$, if $g(x)=r(x)g_1(x)$, then $$\|g\|_{2}^{2}=\sum_{i=0}^{v-1}|r(\xi^{i})g_{1}(\xi^{i})/\sqrt{v}|^2\geq c^2v\prod_{i=0}^{v-1}|g_{1}(\xi^i)|^{2/v}\geq c^2v,$$ where the first inequality follows from the arithmetic-geometric mean inequality. Similarly, when $g(x)=r(x^{v-1})g_2(x)$, then $\|g\|_{2}^2\geq c^2v$ holds for the same reason.  Therefore, from \eqref{eqsmooth}, we have $\matA\geq \eta_{\varepsilon}(R)$.
 \end{proof}
       
 \begin{lemma}\label{BDDtoLWE2}
       Let $R=\Z[G]$ be a group ring with $G=\Z_{u}\ltimes\Z_{v}=\langle s\rangle\ltimes\langle t\rangle$ defined by \eqref{typei} and let $n:=uv$. There exists a probabilistic polynomial-time (classical) algorithm that given an oracle that solves $R\mbox{-}\DLWE_{q,\Psi_{\leq\alpha}}$ and input a number $\alpha\in (0,1)$ and an integer $q\geq 2$, an invertible right ideal $\I$ in $R[G]$ with $\det(\I)$ coprime with $q$, a parameter $r\geq \sqrt{2}q\cdot \eta_{\varepsilon}(\I)$, and polynomially many samples from the discrete Gaussian distribution $D_{\I,\matA}$, where $\matA$ is the matrix representation of any $r(x)\in W_{r,\iota, T}$ (where $\iota=1/\poly(n), T=\poly(n)$) viewed as an element in $\R[G]$, solves $\GDP_{\I^{\vee},g}$ for $g=1/n\cdot \alpha q/(\sqrt{2}r)$.
 \end{lemma}

 \begin{proof} 
       Let $\varphi$ denote the coefficient embedding of the group ring $\R[G]$ into $\R^n$.

       If $\alpha<\exp(-n)$, then except with negligible probability the coset representative $e$ from the instance will satisfy $$\|\varphi(e)\|\leq\sqrt{n}g\leq\alpha/(2\sqrt{n}\cdot\eta_{\varepsilon}(\I))\leq 2^{n}\lambda_1(\I^{\vee}),$$ where the last equality is derived from Lemma \ref{smooth}, which means $\|\varphi(e)\|$ is short enough for us to use Babai's algorithm \cite{DBLP:journals/combinatorica/Babai86} to obtain the solution of the \GDP instance. Hence we can assume $\alpha>\exp(-n)$ without loss of generality. We let $\kappa=\poly(n)$ with $\kappa\geq 100n^2\ell$ such that the advantage of $R\mbox{-}\DLWE$ oracle is at least $2/\kappa$, where $\ell$ is the number of samples required by the oracle.

       We first view $e$ as a  bivariate polynomial with indeterminates $s$ and $t$: $$e(s,t):=f_0(t)+sf_1(t)+\cdots+s^{u-1}f_{u-1}(t)$$ and denote the univariate polynomial $e_j(t):=e(\omega^{j},t)=f_0(t)+\omega^{j} f_1(t)+\cdots+\omega^{j(u-1)}f_{u-1}(t)$. Denote $\rho_{j}(e)=(e_j(0), e_j(\xi), e_j(\xi^2), \ldots, e_{j}(\xi^{v-1})):=(\rho_{j}^{(0)}(e),\rho_{j}^{(1)}(e),\ldots,\rho_{j}^{(v-1)}(e))$ for $j=0,1,\ldots,u-1$, where $\omega$ and $\xi$ are the $u$-th and $v$-th primitive roots of unity, respectively. In the following, we determine $\rho_j$ by determining each of its coordinates $\rho_{j}^{(i)}$ for $i=0,1,2,\ldots,v-1$.        
 
       The reduction uses the decisional $\GRLWE$ oracle to simulate oracles $$\OO_{j}^{(i)}:\C\times\R_{\geq 0}\rightarrow \{0,1\},\quad 0\leq i\leq v-1,$$ such that the probability that $\OO_{j}^{(i)}(z,m)$ outputs 1 only depends on $\exp(m)\left|z-\frac{\rho_{j}^{(i)}(e)+\rho_{u-j}^{(i)}(e)}{2}\right|$, where $z\in\C$ with $\left|z-\frac{\rho_{j}^{(i)}(e)+\rho_{u-j}^{(i)}(e)}{2}\right|$ sufficiently small. Hence, $\OO_j^{(i)}$ serves as an oracle with ``hidden center'' $\frac{\rho_{j}^{(i)}(e)+\rho_{u-j}^{(i)}(e)}{2}$ as defined in Definition \ref{OHCP}. Likewise, we can also use the decision $\GRLWE$ oracle to simulate oracles with ``hidden centers'' $\frac{\rho_{j}^{(i)}(e)-\rho_{u-j}^{(i)}(e)}{2}$. Combining these results, we can retrieve $\rho_{j}^{(i)}(e)$ and $\rho_{u-j}^{(i)}(e)$.
       
       Fix an index $j$ and apply Proposition~\ref{ohcp} to obtain a sufficiently accurate approximation of $\frac{\rho_{j}^{(i)}(e)\pm\rho_{u-j}^{(i)}(e)}{2}$ for each $i$, which in turn enables recovery of $e_j(t)\pm e_{u-j}(t)$ via a system of linear equations.
Having all $e_j$ for every $j$, we can efficiently reconstruct $e$ except with negligible probability.

First, we note that we can efficiently compute $$\frac{e_j(t)\pm e_{u-j}(t)}{2}+\mathcal{I}^{-1}$$ from a representative of $e+\mathcal{I}^{-1}$ for all indices $j$. Consider a representative of the coset $e+\mathcal{I}^{-1}$, written as $e+x$ where $x\in\mathcal{I}^{-1}$. Let $x$ be written as a bivariate polynomial in the indeterminates $s$ and $t$, just as $e$ is: $x(s,t):=x_0(t)+sx_1(t)+\cdots+s^{u-1}x_{u-1}(t)$. As we defined $e_j(t)$ previously, we can define $x_j(t)$ in the same manner. For each $j$, we can compute $e_j(t)+x_j(t)$ by evaluating $e(s,t)+x(s,t)$ at $s=\omega^j$. We claim that $(e_j(t)+e_{u-j}(t)+x_j(t)+x_{u-j}(t))/2$ is a representative of the coset $\frac{e_j(t)+e_{u-j}(t)}{2}+\mathcal{I}^{-1}$, which directly achieves our goal. It remains to show that $(x_j(t)+x_{u-j}(t))/2\in \mathcal{I}^{-1}$. Since $x(s,t)\in \mathcal{I}^{-1}\supseteq \Z[G]$ and $x_j(t), x_{u-j}(t)$ have conjugate coefficients in $\Z[\omega]$, it follows that all the coefficients of $(x_j(t)+x_{u-j}(t))/2$ lie in $\R\cap\Z[\omega]=\Z$, which implies that $(x_j(t)+x_{u-j}(t))/2\in \Z[G]\subseteq \mathcal{I}^{-1}$. The same argument also applies to $\frac{e_j(t)-e_{u-j}(t)}{2}+\mathcal{I}^{-1}$.

       To achieve our goal, when $j=0,u/2$, we can use similar process of Lemma 9 in \cite{cheng2022lwe} to recover $\rho_{j}^{(i)}$ (In this case, $\rho_{j}^{(i)}$ is real). For the following discussion, we may assume $j\neq 0,u/2$. Define $k_j^{(i)}:\C\rightarrow \R[G_1]$ satisfying $\rho_{j}(k_{j}^{(i)}(z))=z\cdot\vece_{i}+\bar{z}\cdot\vece_{v-i}$, where $\vece_{i}$ has 1 in the $i$-th coordinate and 0 otherwise, and we may restrict the image of $k_{j}^{(i)}$ within elements of $\R[G_1]$ which are of the form $$a_0+a_1t+a_2t^2+\cdots+a_{v-1}t^{v-1}\in\R[G_1],$$ which has zero coefficients on $s^it^j$ for any $1\leq i\leq u-1$ and $0\leq j\leq v-1$. On input $(z,m)$, the oracle $\OO_{j}^{(i)}$ uses fresh Gaussian samples from $D_{\I,\matA_{k}^{(i)}}$, where $\matA_{k}^{(i)}$ is the matrix representation of $r_{k}^{(i)}\in W_{r,\iota,T}$ and $(1+\iota)^{k}=\exp(m)$ as in Definition \ref{lagrange}. Then it performs the transformation from Lemma \ref{constructlwe} on these samples, the coset $\frac{e_j+e_{u-j}}{2}-\sum k_{j}^{(i)}(z_{j}^{(i)})+\I^{-1}$, parameter $r$ and matrix norm bound  $d=\alpha q/(\sqrt{2}r)\cdot\omega(1)$, and convert them into $\GRLWE$ samples. Denote these samples by $A^{(i)}_{j,z,m}$. Then $\OO_{j}^{(i)}$ calls the $R\mbox{-}\DLWE$ oracle on these samples and outputs 1 if and only if it accepts.

       Next, the reductions runs the algorithm for each $i=1,2,\ldots, v-1$ with oracle $\OO_j^{(i)}$, confidence parameter $\kappa$, and distance bound $d'=d/(1+1/\kappa)$, and outputs some approximation $z_j^{(i)}$ to the oracle's center. Finally, the reduction runs Babai's algorithm on the coset $\frac{e_j+e_{u-j}}{2}-\sum k_j^{(i)}(z_j^{(i)})+\I^{-1}$, receiving as $\tilde{e}_j^{+}$, and returns $\tilde{e}_j^{+}+\sum k_j^{(i)}(z_j^{(i)})$ as output. 

       The running time of the reduction is polynomial time in the size of the group ring. Assuming $z_j^{(i)}$ are valid solution to $(\exp(-\kappa),\exp(-\kappa),1+1/\kappa)$-OHCP with hidden center $\frac{\rho_{j}^{(i)}(e)+\rho_{u-j}^{(i)}(e)}{2}$, we check that the correctness of the algorithm. Since $z_j^{(i)}$ are valid solutions, we have $$\left|z_j^{(i)}-\frac{\rho_{j}^{(i)}(e)+\rho_{u-j}^{(i)}(e)}{2}\right|\leq\exp(-\kappa)d'\leq\exp(-\kappa)/\eta(\I)\leq 2^{-n-1}\lambda_1(\I^{-1})/\sqrt{n}$$ by the definition of OHCP.
       Thus, $\left\|\sum k_j^{(i)}(z_j^{(i)})-\frac{e_j(t)+e_{u-j}(t)}{2}\right\|\leq 2^{-n}\lambda_{1}(\I^{-1})$. Note that $e_j(t)$ and $e_{u-j}(t)$ have conjugate coefficients at each position, it follows that $\frac{e_j(t)+e_{u-j}(t)}{2}$ can be regarded as an element in $\R[G]$. The Babai's algorithm will return the exact value of $\frac{e_j(t)+e_{u-j}(t)}{2}-\sum k_j^{(i)}(z_j^{(i)})$, which we denote by $\tilde{e}^{+}$, then finally we output $\frac{e_j+e_{u-j}}{2}=\tilde{e}^{+}+\sum k_j^{(i)}(z_j^{(i)})$. The analysis also applies to $\frac{e_j-e_{u-j}}{2}$, which in turn gives the value of $e_{j}$ and $e_{u-j}$. Hence we prove the correctness of the algorithm.

       It remains to prove, except with negligible probability over the choice of $e$ and for all $i,j$:

       \begin{enumerate}
              \item [(1)]  $\OO_j^{(i)}$ represents valid instances of $(\exp(-\kappa),\exp(-\kappa),1+1/\kappa)$-OHCP with ``hidden center'' $\frac{\rho_{j}^{(i)}(e)+\rho_{u-j}^{(i)}(e)}{2}$;
              \item [(2)] $\OO_j^{(i)}$ satisfies the condition of what is stated in Proposition \ref{ohcp}.  
       \end{enumerate}

       To prove validity, we first observe that the distribution $A^{(i)}_{j,z,m}$ depends only on $\exp(m)\left|z-\frac{\rho_{j}^{(i)}(e)+\rho_{u-j}^{(i)}(e)}{2}\right|$ if $\left|z-\frac{\rho_{j}^{(i)}(e)+\rho_{u-j}^{(i)}(e)}{2}\right|\leq (1+1/\kappa)d'=d$. We have $$\exp(-\kappa)d'\leq \exp(-n)d\leq \left|\frac{\rho_{j}^{(i)}(e)+\rho_{u-j}^{(i)}(e)}{2}\right|\leq d'.$$ Therefore, $\OO_{j}^{(i)},\kappa,d'$ correspond to a valid instance of $(\exp(-\kappa),\exp(-\kappa),1+1/\kappa)\mbox{-}$OHCP with ``hidden center'' $\frac{\rho_{j}^{(i)}(e)+\rho_{u-j}^{(i)}(e)}{2}$, except with negligible probability.

       Finally, we prove that the oracle $\OO_{j}^{(i)}$ indeed satisfies the three conditions specified in Proposition \ref{ohcp}.
       \begin{enumerate}
              \item [(1)] For fixed $j$, denote $p_{j}^{(i)}(z,m)$ as the probability that $\OO_{j}^{(i)}$ outputs 1 on input $(z,m)$ and $p^{(i)}_\infty$ for the probability that $R\mbox{-}\DLWE$ oracle outputs 1 on uniformly random inputs. It follows that $p_{j}^{(i)}(0,0)=p_{j'}^{(i')}(0,0)$ for all $(i,j), (i',j')$, and $p_{j}^{(i)}(0,0)-p^{(i)}_{\infty}$ is exactly the advantage that $R\mbox{-}\DLWE$ oracle has against the error rate that we derive from the transformation described by Corollary \ref{groupsmoo}. Recall that $e$ is drawn from $D_{d/n}$, then $\|e\|_{\operatorname{Mat}}$ is no more than $d$. Similar to the procedure in Lemma \ref{constructlwe}, the resulting $R\mbox{-}\LWE$ samples are exactly distributed from $\Psi_{\leq\alpha}$. Since the decisional $\GRLWE$ oracle has an advantage $2/\kappa$ against this distribution of error rate. By Markov's inequality, we may assume that $p_{j}^{(i)}(0,0)-p_{\infty}^{(i)}\geq 1/\kappa$ holds with non-negligible probability. It means Item 1 in Proposition \ref{ohcp} is satisfied.

\item [(2)] For Item 2, the distribution of $A^{(i)}_{j,0,m}$ is within negligible statistical distance of distribution $A_{s,A_{k}^{(i)}}$. Recall that $r_{k}^{(i)}(\xi^{i})=r_{k}^{(i)}(\xi^{v-i})=r(1+\iota)^{k}$ and $r_{k}^{(i)}(\xi^{h})=r$ for $h\neq i,v-i$.
With Lemma \ref{largesmooth} and the definition of the smoothness parameter, we can prove the distribution of $A^{(i)}_{j,0,m}$ is within statistical distance 
\begin{align}\label{item2} \ell\exp\left(-v\prod_{h}\left(1/\sqrt{v}\cdot r_{k}^{(i)}(\xi^{h})\right)^{2/v}\right)&\leq \ell\exp\left(-\exp(4m/v)\cdot \left(r/q\right)^2\cdot\prod_{i}\rho_{j}^{(i)}(e_{j}+e_{m-j})^{2/v}\right)\nonumber\\
&\leq \ell\exp(-\exp(4m/v-4n-1))\nonumber\\
&\leq 2\exp(-m/\kappa)       
\end{align}
of the uniform distribution. Here we use $$|\rho_{j}^{(i)}(e_{j}+e_{m+j})|\geq \exp(-n)d>\exp(-n)\cdot \alpha q/(\sqrt{2}r)>\exp(-2n-1/2)\cdot q/r.$$ The inequality of \eqref{item2} follows from the fact that $\exp(4m/v-4n-1)\gg m/\kappa +\log(\ell/2)$. Therefore, we can conclude that $|p_{j}^{(i)}(0,m)-p_{\infty}^{(i)}|\leq 2\exp(-m/\kappa)$, which means Item 2 in Proposition \ref{ohcp} is satisfied.
       
\item [(3)] By Lemma \ref{Gaussianbound}, the distribution of $A^{(i)}_{j,z,m_1}$ and $A^{(i)}_{j,z,m_2}$ are within statistical distance $$\min\{1,10\ell(\exp(|m_1-m_2|)-1)\}\leq \kappa|m_1-m_2|, $$where $\ell$ is the number of the samples we have used as mentioned before. Thus, we have proved $p_{j}^{(i)}(z,m)$ is $\kappa$-Lipschitz.
\end{enumerate}

Here we complete the proof.
\end{proof}
    
\section{Summary}
\label{sec:5}
\subsection{Relation to Module-\LWE}
Similar to the discussion in \cite{cheng2022lwe}, our construction is closely related to the Module-\LWE problems. As an illustrative example, consider the construction arising from Type I groups. Let $R^{(1)}$ be the group ring associated with the Type I group $G_1=\Z_m\ltimes \Z_n=\langle s\rangle \ltimes \langle t\rangle$ specified in Section 3, and let $R_q^{(1)}=R^{(1)}/qR^{(1)}$. Define $\mathfrak{s}=\sum_{i=0}^{m-1}f_i(t)s^i$ and $\mathfrak{a}=\sum_{i=0}^{m-1}g_i(t)s^i$, where each $f_i(t),g_i(t)\ (0\leq i\leq m-1)$ is a polynomial over $\Z[t]/\langle t^{n/2}+1\rangle$ with coefficients sampled uniformly at random. Let $\mathfrak{b}=\sum_{i=0}^{m-1}b_i(t)s^i$ and let $\mathfrak{e}=\sum_{i=0}^{m-1}e_i(t)s^i$ be an error element in $R_{q}^{(1)}$. By the definition of the multiplication operation in the group ring, the equation $\mathfrak{b}=\mathfrak{s}\cdot\mathfrak{a}+\mathfrak{e}$ can can be written in matrix form: $$\left(\begin{array}{c}
       b_0\\
       \vdots\\
       b_{m-1}       
\end{array}\right)=\matM\cdot\left(\begin{array}{c}
       g_0\\
       \vdots\\
       g_{m-1}       
\end{array}\right)+\left(\begin{array}{c}
       e_0\\
       \vdots\\
       e_{m-1}       
\end{array}\right),$$ where each of the entries in $\matM$ is a polynomial in $\Z[t]/\langle t^{n/2}+1\rangle$ that can be efficiently computed from $f_0,\ldots, f_{m-1}$. Additionally, our construction can be viewed as a structured variant of the Module-\LWE problem. This structure contrasts with the standard Module-\LWE setting proposed in \cite{brakerski2014leveled}, in which the corresponding matrix $\matM$ is typically assumed to be unstructured, with entries sampled independently and uniformly at random. In our case, the $m^2$ entries of $\matM$ are fully determined by only $m$ polynomials. Consequently, to generate $O(mn)$ pseudorandom samples, our construction requires only $O(mn)$ random elements, namely, all the coefficients of $f_0,\ldots,f_{m-1}, g_0,\ldots, g_{m-1}$, achieving a reduction by a factor of $m$ compared to the standard Module-\LWE problem. The construction in \cite{cheng2022lwe} can be viewed as a special case of our framework when $m=2$, whereas our approach allows more flexible parameter choices. 
Moreover, according to the reduction result in \cite{albrecht2017large}, an instance of our rank-$m$ Module-\LWE problem can be reduced to an instance of the Ring-\LWE problem over a cyclotomic ring with modulus $q^m$.

\subsection{Conclusion}
Both the search and decisional version of the \LWE problem over group rings in this paper enjoy the hardness due to the reductions from some computationally hard problems in ideal lattices. Specifically, we focus on the finite non-commutative groups constructed via the semi-direct product of two cyclic groups (in some sense, a group family that owns the simplest structure). While there are indeed various methods for constructing non-commutative groups from two smaller groups, we believe the results of this paper can be generalized to groups with more complex structures. In fact, the two types of groups discussed in this paper are special cases of metacyclic groups, i.e., any semi-direct product of two cyclic groups (which may not necessarily be induced by $\varphi$ and $\psi$ mentioned in the definition of \eqref{typei} and \eqref{typeii}).

Additionally, the process of reduction extensively exploits properties of the irreducible representations and their eigenvalues. On the one hand, they are closely related to the left regular representation of the ring elements, which encompasses all irreducible subrepresentations of them to some extent. On the other hand, as mentioned in \cite{cheng2022lwe}, the dimension of irreducible representations affects the efficiency of computing the multiplication of two elements of the group rings. Therefore, it is essential for the dimension of the irreducible representations of the group elements to be not exceedingly large. In this paper, we have chosen two types of rings that have been utilized to implement the \LWE problem. These rings are quotient rings obtained by group rings modulo an ideal. Nevertheless, these selections are made heuristically such that the resulting rings would not suffer from the potential attack exploiting the one-dimensional subrepresentations of the group ring. However, it remains an open problem whether there is any general approach to selecting quotient rings that are not vulnerable to such kind of attack. 

\bibliographystyle{IEEEtran}
\bibliography{reference}

\begin{appendices}
 \section{Proof of Lemma \ref{lemtypei}}
      \begin{proof}
        (i) When $m$ is even, by calculating the number of conjugate classes of $G_1$, we know the number of irreducible representations of $G_1$ is $st+2s$.
      
        To determine all 1-dimensional irreducible representations of $G_1$, it suffices to consider the quotient group $G_1/[G_1,G_1]$, where $[G_1,G_1]$ deontes the commutator group of $G_1$, which is isomorphic to $\Z_n=\langle t\rangle$. Then $G_1/[G_1,G_1]\cong \Z_m$ is also a cyclic group. According to the representation theory, the number of 1-dimensional representations of $G$ and $\Z_m$ are identical, i.e., $m=2r$. Furthermore, $\chi_{i}(s)$ is completely determined by all the (1-dimensional) representations of $\Z_m$, as stated in the lemma. Since $t$ and $st$ belong to the same conjugate class, then $\chi_i(t)=\chi_{i}(st)/\chi_{i}(s)=1$.
      
        It remains to determine all 2-dimensional irreducible representations of $G_1$. Let $(\rho, V)$ be a 2-dimensional irreducible representation of $G_1$. We next consider the restriction of $\rho$ over the subgroup $\Z_n$. Since $\Z_n$ is an Abelian group, then $V$ can be decomposed into two 1-dimensional subrepresentations $V=V_1\oplus V_2$ when viewed as a representation of $\Z_n$. It is important to emphasize that $V_1$ and $V_2$ are distinct. Since $$\rho(t)(\rho(s)V_1)=\rho(s)\rho(s^{-1}ts)V_1\in \rho(s)V_1,$$ $\rho(s)V_1$ is also a subrepresentation of $\Z_n$. It has to be $\rho(s)V_1=V_2$ and similarly $\rho(s)V_2=V_1$. Fixing a nonzero vector $v_1\in V_1$, then there exists $\lambda\in \C$, such that $\rho(t)v_1=\lambda v_1$. Let $v_2=\rho(s)v_1\in V_2$, and then there also exists $\zeta\in\C$, such that $$\rho^2(s)v_1=\rho(s)v_2=\zeta v_1.$$ As $s$ has order $m=2r$, it follows that $\zeta$ is an $r$-th root of unity. Thus, the transformation matrix of $\rho(s)$ with respect to the basis $v_1,v_2$ is given by $$\left(\begin{array}{cc}
         0& \zeta\\
         1& 0      
        \end{array}\right).$$ On the other hand, we notice that $$\rho(t)v_2=\rho(t)\rho(s)v_1=\rho(s)\rho(t^{-1})v_1=\bar{\lambda}\rho(s)v_1,$$ where $\bar{\lambda}$ is the complex conjugate of $\lambda$. If $\lambda$  is real, then $\rho^2(t)=\operatorname{id}$. In this case, $G_1/\ker(\rho)\cong\Z_m\ltimes \Z_2=\Z_m\times\Z_2$ is abelian, which implies it has no 2-dimensional irreducible representations, hence it makes a contradiction. Then $\lambda$ has to be the nonreal root of $x^n-1$. Hence the transformation matrix of $\rho(t)$ with respect to the basis $v_1,v_2$ is $$\left(\begin{array}{cc}
               \xi^{j}& 0\\
              0 & \xi^{n-j}      
              \end{array}\right),\quad j=1,2,\ldots,u.$$
        
        (ii) It is essentially the same as the discussion in (i). 
      \end{proof}
      
      \section{Proof of Lemma \ref{dual ideal}}
      \begin{proof}
        \textbf{Type I}: Let $m$ be an even positive integer and $n$ an arbitrary positive integer, and let $G_1$ be a group defined as $$\mathbb{Z}_m\ltimes_{\varphi} \mathbb{Z}_n=\langle s,t\mid s^m=1,t^n=1,sts^{-1}=t^{-1}\rangle,$$ where $s$ and $t$ are the generators of cyclic groups $\Z_m$ and $\Z_n$, respectively.
        
        The group ring $\Z[G_1]$ naturally has a basis over $\Z$, which are exactly all the elements of $G_1$, i.e.,
        \begin{align}\label{basis}
               \{s^{i}t^{j}\mid 0\leq i\leq m-1,0\leq j\leq n-1\}.
        \end{align}
      Let $I$ be a right ideal in $\Z[G]$, and denote the ideal lattice corresponding to $I$ under coefficient embedding by $\I$. Suppose that $$\mathfrak{h}_{1}=\sum_{j=0}^{n-1}x_{0,j}t^j+\sum_{j=0}^{n-1}x_{1,j}st^{j}+\cdots+\sum_{j=0}^{n-1}x_{m-1,j}s^{m-1}t^{j}\in I^{-1}\subseteq \R[G_1],$$ where $x_{i,j}\in\mathbb{R}$, then $\mathfrak{h}_1$ can be regareded as a $mn$-tuple with real components under coefficient embedding, i.e., $$\vecx:=(\vecx_0,\vecx_1,\vecx_2,\ldots,\vecx_{m-1})\in \I^{\vee}\subset \R^{mn},$$ where $\vecx_{i}:=(x_{i,0},x_{i,1},\ldots,x_{i,n-1})$. Denote $\tilde{\vecx}_{i}:=(x_{i,0},x_{i,n-1},x_{i,n-2},\ldots,x_{i,2},x_{i,1})$, let \begin{align*}
               \vecz&:=(z_{0,0},z_{0,1},\ldots,z_{0,n-1},z_{1,0},z_{1,1},\ldots,z_{1,n-1},\ldots,z_{m-1,0},\ldots,z_{m-1,n-1})\\
               &=(\tilde{\vecx}_{0},\vecx_{m-1},\tilde{\vecx}_{m-2},\vecx_{m-3},\ldots,\tilde{\vecx}_{2},\vecx_{1}),
        \end{align*}
        then $\vecz$ is also an $mn$-tuple obtained by permutating the coordinates of $\vecx$.
      
        It suffices to show that $\vecx\in\I^{-1}$ if and only if $\vecz\in\I^{\vee}$. For simplicity, the addition and multiplication operations with respect to the first (or second) subscript of $x,y,z$ are all modulo $m$ (or $n$) (also the same in the proof for Type II).
      
        Note that $\vecx\in \I^{-1}$ if and only if, for any element $\mathfrak{h}=\sum_{i=0}^{m-1}\sum_{j=0}^{n-1}y_{i,j}s^{i}t^{j}\in I$, we have $\mathfrak{h}_1\mathfrak{h}\in\Z[G_1]$. By computing $\mathfrak{h}_1\mathfrak{h}$ represented by the basis \eqref{basis}, we obtain that $\vecx\in \I$ is equivalent to the condition that for any $a=0,1,\ldots,m-1,b=0,1,\ldots,n-1$, we have $$\sum_{j=0}^{n-1}x_{0,j}y_{a,b-(-1)^aj}+\sum_{j=0}^{n-1}x_{1,j}y_{a-1,b+(-1)^aj}+\cdots+\sum_{j=0}^{n-1}x_{m-1,j}y_{a+1-m,b+(-1)^aj}\in\mathbb{Z}$$ holds. Applying this result to $\vecz$, we get  
        \begin{align}\label{inverse}
               \sum_{j=0}^{n-1}z_{0,j}y_{a,b+(-1)^aj}+\sum_{j=0}^{n-1}z_{m-1,j}y_{a-1,b+(-1)^aj}+\cdots+\sum_{j=0}^{n-1}z_{1,j}y_{a+1-m,b+(-1)^aj}\in\mathbb{Z}.
        \end{align}
      
        On the other hand, if $\vecz\in\I^{\vee}$, then for any element $\sum_{i=0}^{m-1}\sum_{j=0}^{n-1}y_{i,j}s^{i}t^{j}\in I$, we have $$\sum_{i=0}^{m-1}\sum_{j=0}^{n-1}z_{i,j}y_{i,j}\in\Z.$$ Since $I$ is a right ideal of $\Z[G]$, \begin{align*}
        \left(\sum_{i=0}^{m-1}\sum_{j=0}^{n-1}y_{i,j}s^{i}t^{j}\right)s^{m-a}t^{(-1)^ab}=\sum_{i=0}^{m-1}\sum_{j=0}^{n-1}y_{i,j}s^{m-a+i}t^{-(-1)^a(j-b)}\in I
        \end{align*} holds. Then we have $\vecz\in\I^{\vee}$ if and only if for any $a\in[m],b\in[n]$, $$\sum_{i=0}^{m-1}\sum_{j=0}^{n-1}z_{m-a+i,(-1)^a(j-b)}y_{i,j}=\sum_{i=0}^{m-1}\sum_{j=0}^{n-1}z_{m-a+i,j}y_{i,b+(-1)^aj}\in\Z,$$ which is exactly the same as \eqref{inverse}. Hence we have proved the claim with respect to groups of Type I.
        
      \textbf{Type II}: Let $n>2$ be an integer. We define $G_2$ as $\Z_{n}^{\ast}\ltimes_{\psi} \Z_{n}$, i.e., $$G_2:=\{a\odot g^{k}\mid a\in\Z_{n}^{\ast},0\leq k<n-1\},$$ where $g$ is a generator of $\Z_n$ (with order $n$). The multiplication of $G_2$ is induced by $\psi$, i.e., for any $a,b\in\Z_{n}^{\ast}$, and $k_1,k_2\in\Z$, $$(a\odot g^{k_1})(b\odot g^{k_2})=ab\odot g^{k_1b^{-1}+k_2}.$$
      Then the group ring $\Z[G_2]$ has a natural $\Z$-basis:$$\{a\odot g^{k}\mid a\in\mathbb{Z}_{n}^{\ast}, k=0,1,2,\ldots,n-1\}.$$ 
      
      Since $m:=\varphi(n)$ is even when $n>2$, we can write $$\Z_{n}^{\ast}=\left\{a_0=1,a_1,a_2,\ldots,a_{m/2-1},a_{m/2}=m-1,a_{m/2+1},\ldots,a_{m-1}\right\},$$ where $a_{m-i}=a_{i}^{-1}$ holds for any integer $i$ (by arranging the order of the elements in $\Z_{n}^{\ast}$ appropriately). Let $\sigma:\Z_{n}^{\ast}\rightarrow [m]$ be an index mapping, which maps each element of $\Z_{n}^{\ast}$ to its index, i.e., $\sigma(a_i)=i$ for any $i\in [m]$.
      
      Suppose that $\mathfrak{a}_1=\sum_{j=0}^{n-1}x_{0,j}(a_0\odot g^{j})+\sum_{j=0}^{n-1}x_{1,j}(a_1\odot g^{j})+\cdots+\sum_{j=0}^{n-1}x_{m,j}(a_{m-1}\odot g^{j})\in \R[G_2]$. Similar to the proof for Type I, $\mathfrak{a}_1$ can also be regarded as an $mn$-tuple $\vecx$ under coefficient embedding. Let $\vecz$ be another $mn$-tuple satisfying $$z_{i,j}=x_{m-i,-ja_i}.$$ Since $\gcd(a_i,n)=1$, it can be easily verified that $\vecz$ can be obtained by permutating the coordinates of $\vecx$. 
      
      We claim that $\vecx\in\I^{-1}$ if and only if $\vecz\in\I^{\vee}$ as in the case of Type I.
      
      We have $\mathfrak{a}_1\in I^{-1}$ if and only if for any $\mathfrak{a}=\sum_{i=0}^{m-1}\sum_{j=0}^{n-1}y_{i,j}(a_i\odot g^{j})\in I$, $\mathfrak{a}_1\mathfrak{a}\in\Z[G]$. It is equivalent to that for any integer $r,s$, the coefficient of $\mathfrak{a}_2\mathfrak{a}$ over $a_r\odot g^s$ is also an integer, i.e.,\begin{align}
        &\sum_{j=0}^{n-1}x_{r,s-j}y_{0,j}+\sum_{j=0}^{n-1}x_{\sigma(a_ra_{1}^{-1}),(s-j)a_1}y_{1,j}+\cdots+\sum_{j=0}^{n-1}x_{\sigma(a_ra_{m-1}^{-1}),(s-j)a_{m-1}}y_{m-1,j} \nonumber\\
        =&\sum_{j=0}^{n-1}z_{\sigma(a_0a_r^{-1}),(j-s)a_r}y_{0,j}+\sum_{j=0}^{n-1}z_{\sigma(a_1a_r^{-1}),(j-s)a_r}y_{1,j}+\cdots+\sum_{j=0}^{n-1}z_{\sigma(a_{m-1}a_r^{-1}),(j-s)a_r}y_{m-1,j}\in\Z.\label{inverse2}
      \end{align}
      
      On the other hand, if $\vecz\in\I^{\vee}$, then $$\langle\vecz,\vecy\rangle=\sum_{i=0}^{m-1}\sum_{j=0}^{n-1}z_{i,j}y_{i,j}\in\Z.$$ Since $\I$ is a right ideal,  for any $r=0,1,2,\ldots,m-1,s=0,1,2,\cdots,n-1,$ we have $$\left(\sum_{i=0}^{m-1}\sum_{j=0}^{n-1}y_{i,j}\left(a_i\odot g^j\right)\right)\left(a_r^{-1}\odot g^{-sa_r}\right)=\sum_{i=0}^{m-1}\sum_{j=0}^{n-1}y_{i,j}\left(a_ia_r^{-1}\odot g^{(j-s)a_{r}}\right)\in I.$$
      Thus, we could claim $\vecz\in\I^{\vee}$ if and only if for any $r\in [m],s\in [n]$, $$\sum_{i=0}^{m-1}\sum_{j=0}^{n-1}z_{\sigma(a_{i}a_r^{-1}),(j-s)a_r}y_{i,j}\in\Z,$$ which is exactly the same as \eqref{inverse2}.
      \end{proof}
\end{appendices}

\end{document}